\documentclass[journal]{IEEEtran}

\usepackage{graphicx}
\usepackage{subfig}

\usepackage{epsf}
\usepackage{amssymb}
\usepackage{color,ulem}
\usepackage{amsthm}

\newtheorem{lem}{Lemma}
\newtheorem{thm}{Thorem}

\begin{document}

\title{The Geometric Unscented Kalman Filter}

\author{Chengling Fang, Jiang~Liu, Songqing~Ye  and  Ju Zhang
\thanks{Chengling Fang is with Chongqing Jiaotong University.}
\thanks{Jiang Liu, Songqing~Ye  and  Ju Zhang are with Chongqing Institute of Green and Intelligent Technology, Chinese Academy of Sciences.}}

%\thanks{Manuscript received April 19, 2015; revised September 17, 2015.}

\maketitle

% As a general rule, do not put math, special symbols or citations  in the abstract or keywords.
\begin{abstract} Many filters have been proposed in recent decades for the nonlinear state estimation problem. The linearization-based extended Kalman filter (EKF) is widely applied to nonlinear industrial systems. As EKF is limited in accuracy and reliability, sequential Monte-Carlo methods or particle filters (PF) can obtain superior accuracy at the cost of a huge number of random samples.  The unscented Kalman filter (UKF) can achieve adequate accuracy more efficiently by using deterministic samples, but its weights may be negative, which might cause instability problem. For Gaussian filters, the cubature Kalman filter (CKF) and Gauss Hermit filter (GHF) employ cubature and respectively Gauss-Hermite rules to approximate statistic information of random variables and exhibit impressive performances in practical problems. Inspired by this work, this paper presents a new nonlinear estimation scheme named after geometric unscented Kalman filter (GUF). The GUF chooses the filtering framework of CKF for updating data and develops a geometric unscented sampling (GUS) strategy for approximating random variables. The main feature of GUS is selecting uniformly distributed samples according to the probability and geometric location similar to UKF and CKF, and having positive weights like PF. Through such way, GUF can maintain adequate accuracy as GHF with reasonable efficiency and good stability. The GUF does not suffer from the exponential increase of sample size as for PF or failure to converge resulted from non-positive weights as for high order CKF and UKF.
\end{abstract}

\begin{IEEEkeywords}
Gaussian Filter, Nonlinear Estimation, Cubature Kalmam Filter, Unscented Kalman Filter, Particle  Filter.
\end{IEEEkeywords}

\IEEEpeerreviewmaketitle

\section{Introduction}

\IEEEPARstart {N}{onlinear} filtering has been widely studied in many science and engineering disciplines. Bayesian estimation theory provides a general filtering framework for this \cite{jazwinski2007stochastic}, which utilizes Bayes' rule to estimates the probabilistic state of a system. Computing the posterior probability density function (PDF) is a crucial part of this method. However, multidimensional integrals are typically intractable \cite{OptimalFiltering}, and a closed-form solution to the posterior density is available only for a restricted class of filters. For example, if the dynamic state-space model is linear with additive Gaussian noise and the prior distribution of the state
variable is Gaussian, then the well-known Kalman filter (KF) \cite{kalman1961new} provides a closed-form solution. For general cases, various approximate methods were proposed to estimate PDF. These methods are categorized into two classes: global and local methods \cite{arasaratnam2007discrete, 2009Cubaturekalmanfilters, vsimandl2009derivative}.

The global approach makes no explicit assumption about prior and posterior's PDF and can achieve satisfactory accuracy with a heavy computational load. This type filters compute the posterior PDF directly by using approximating techniques \cite{2009Cubaturekalmanfilters, vsimandl2009derivative}, for example, the point-mass filter  \cite{vsimandl2006advanced}, the Gaussian mixture filter \cite{alspach1972nonlinear}, the particle filter (PF) \cite{1993Novelapproachtononlinear},
and Quasi-Monte Carlo filter \cite{guo2006quasi}. In general, global methods may have more computational demands than local methods \cite{vsimandl2009derivative}. For example, the PF reformulates the PDF of state with a set of weighted random samples, rather than the function of in the state-space model \cite{Bradley1992A}. As the samples increase, the PDF can be approximated ever more accurately. However, the computational complexity increases exponentially with the dimensions of the system states \cite{Liu98sequentialmonte}, \cite{Khan2004}, \cite{2009Atutorialonparticlefiltering}.  Besides, the performances of PF depend highly on the selection of proposal distributions \cite{guo2006quasi}.  To address such problems, many sampling strategies have been proposed, such as importance sampling (IS), stratified sampling and systematic sampling  \cite{UseofdifferenMontCarlosamplingtechniques}, \cite{OnsequentialMonteCarlosamplingmethodsforBayesianfiltering}, \cite{MonteCarlosamplingmethodsusingMarkovchainsandtheirapplications}. The IS \cite{Robert2005MCS} is the most wildly used since it is easy to implement.
For the sake of computational efficiency, many improvements \cite{Khan2004, Fox03KLD-sampling, Kluwer2003, Khan2004Rao-Blackwell}  have been developed.   In practice, the more nonlinear or non-Gaussian the problem is, the more potential PF would demonstrate, especially when computational power is rather cheap and the data dimension is fairly low  \cite{Gustafsson10particlefilter}.

Under explicit assumption about PDF, the local methods are based on specific approximations of PDF or the nonlinear functions in the state-space model \cite{2009Cubaturekalmanfilters} so that the filtering framework of KF can be used for the Bayesian estimation. The {\it extended Kalman filter} (EKF) \cite{OptimalFiltering}, based on function approximation, is probably the earliest and widely used local method for nonlinear industrial systems. The EKF is computationally efficient. However, it faces two well-known limitations. First, the linearization assumes the existence of the Jacobian matrix. However, this is not always true in practice \cite{Tugnait1982, Kuchar-review, UKF2001}. Second, the linear approximation is only reliable if the remainder of the nonlinear parts is negligible errors. Otherwise, the propagation errors could increase rapidly to severe vibration and divergence  \cite{Athans1968, Mehra1971, Austin1981, Lerro93}. Accordingly, there are various improvements upon EKF as seen in  \cite{TheStabilityAdaptiveTwo-stageExtendedKalmanFilter}, \cite {Amodifiedtemporalapproachtometa-optimizing} and \cite{ExtendedKalmanFilterwithFuzzyMethod}. Their robustness and stability have been discussed in  \cite{RobustExtendedKalmanFiltering}, \cite{Boutayeb1997Convergence}. The EKF was extended to the central difference filter (CDF) \cite{2000Gaussianfiltersfornonlinear} and the divided difference filter (DDF) \cite{schei1997finite, norgaard2000new} without the demand of the Jacobian matrix. They are based on interpolation formula using the similar deterministic sampling approach to approximate the integrand.

In the recent decades, there arose many local filters  based on the polynomial interpolation or PDF approximation: the unscented Kalman filter (UKF) \cite{1997AnewextensionKalman}, the cubature Kalman filter (CKF) \cite{2009Cubaturekalmanfilters},  the Gauss-Hermit filters  (GHF) \cite{kushner2000nonlinear,2008GeneralizedGausHermitefiltering,vsimandl2009derivative}. When the PDF is Gaussian, there are a series of Gaussian approximated (GA) filters based on deterministically chosen weighted points. Besides UKF, CKF and GHF, GA filters include sparse-grid quadrature nonlinear filter (SGQNF) \cite{jia2012sparse}, spherical simplex-radial cubature Kalman filter (SSRCKF) \cite{wang2014spherical}, interpolatory cubature Kalman filter (ICKF) \cite{zhang2015interpolatory}, embedded cubature Kalman filter (ECKF) \cite{zhang2015embedded}, Gaussian sum filters  \cite{Wu2006,GSF2008,AGSF2011}, stochastic integration filter (SIF) \cite{dunik2013stochastic}. These filters can be modified to capture high-order moments by some proper selections of weighted samples to approximate Gaussian PDFs  \cite{1997AnewextensionKalman, 2009Cubaturekalmanfilters}. The UKF and CKF often achieve higher accuracy than EKF with similar complexity  \cite{2004Unscentedfiltering}. Compared with PF, they often achieve high accuracy with a less number of samples, especially for high dimensional systems. Unfortunately, with dimensions increasing, the accuracy of UKF become unreliable  \cite{2009Cubaturekalmanfilters, Wu2006}.  To enhance the accuracy, several improvements were proposed,  such as scaled UKF \cite{2002Thescaledunscented}, high-order unscented filter \cite{TheHigherOrderUnscentedFilter, Grothe2013HoUKF, ALinearExtensionLiuJiang,AHeuristicSigmaLiuJiang}, truncated UKF \cite{TruncatedUnscentedKalmanFiltering}.  However,  such improvements inevitably result in negative weights when the dimension is greater than three, which is probably why the corresponding filters are not reliable or even divergent. The CKF can be regarded as a special case of UKF with a special parameter $\kappa=0$, although it is derived from a different philosophy. It directly estimates the integral  $I(\mathbf f)=\int_{\mathbb R^ n} \mathbf{f}(\mathbf{x})\times \mbox{exp}(-\mathbf{x}^{T}\mathbf{x})\mbox{d}\mathbf{x}$ based on the Cubature rule,  where $\mathbf{f}$ is a  nonlinear function \cite{SphericalRadialIntegration-RulesforBayesianComputation}, \cite{ComparisonMethodsComputationMultivariatetProbabilities}. Similarly, the negative weights still appear in the high order CKF \cite{2013High-degreecubatureKalman}, and might cause the unstable phenomena of this filter. For more discussions on the convergence and improvement of CKF, see \cite{Convergenceanalysisofnon-linearfiltering}, \cite{Robustsquare-rootcubatureKalman} and the references therein. The computational complexity of GHF also grows exponentially with the state dimension. So the computational load is usually prohibitive even for moderately high dimensional dynamical systems. This led to some improvement study \cite{jia2012sparse} of GHF for efficiency. The SIF bases on the stochastic integral rule (SIR) and can eliminate systematic errors caused by nonlinear approximation. Due to the negative weight in 3rd-SIR, the numerical filtering stability cannot be ensured, and the filtering accuracy will degrade greatly \cite{zhang2014quasi}.

Roughly, those mentioned above existing typical nonlinear filtering methods including PF, UKF, CKF, GHF, and 3rd-SIF cannot simultaneously address
numerical instability problem, accuracy and efficiency problems. To simultaneously address these problems, a geometric unscented rule (GUR) is proposed in this article inspired by sampling strategies of PF, UKF and CKF. The major feature of GUR is selecting samples geometrically uniformly distributed on a series of spheres with positive weights.  Then a novel geometric unscented filter (GUF) is obtained by applying the GUR to compute the multidimensional integrals involved in filters. The GUF address the instability problem by positive weights and ensure the accuracy and efficiency by deterministic samples capturing the moments of random variables.
To illustrate the superiority of the proposed GUF algorithm, we present some numerical simulations about target tracking with moderate dimension and high nonlinearity.  As can be seen from simulation results, the new GUF has higher accuracy and better stability than existing filtering algorithms. The efficiency of GUF is confirmed by the comparison result of time-consuming with other methods on the same platform in simulation.

The remainder of this article is organized as follows. In the next section, we briefly review the sampling strategies in PF, UKF and CKF. Based on them, section III presents the novel
nonlinear estimation GUF. Under the framework of GUF, we
study the Gaussian GUF in section IV. Then the numerical simulation
and analysis are given in section V. The last section VI is
composed of some concluding remarks.

\section{Sampling Strategies Review}\label{sec:SSR}
Sampling strategies play a crucial role in the nonlinear filters PF, UKF, GHF and CKF. This section gives a concise review on them. As the sampling takes place in the filters,  we first recall the nonlinear Kalman filtering frame.  This article considers the following model of nonlinear dynamic system:
\begin{eqnarray}
\label{Eq:pmf}
\mathbf{x}_{k+1}&=&\mathbf{f}(\mathbf{x}_{k})+\mathbf{v}_{k}\\
\mathbf{z}_{k+1}&=&\mathbf{h}(\mathbf{x}_{k+1})+\mathbf{w}_{k+1}
\end{eqnarray}
where $\mathbf{x}_{k}\in\mathbb{R}^{n}$; $\mathbf{z}_{k}\in\mathbb{R}^{m}$; $\mathbf{v}_{k}$ and $\mathbf{w}_{k+1}$ are independent Gaussian white process noise and measurement noise with the covariance $\mathbf{Q}_{k}$ and $\mathbf{R}_{k+1}$,  respectively.

Let  $\bar \mathbf{x}_{k|k}$ and $\mathbf{P}_{k|k}$ respectively denote the estimates of mean and covariance of the system state $\mathbf{x}_{k}$ at time $k$. Let  $\mathcal{S}_{i,k|k}$ and  $\omega_{i,k|k}$ be the samples and associated weights, respectively, computed by some sampling strategy for random variable $\mathbf{x}_{k}$, with $1\leq i\leq N$. The sampling $(\mathcal{S}_{i,k|k},\omega_{i,k|k})$ approximate $\mathbf{x}_{k}$ in terms of  the mean  $\bar \mathbf{x}_{k|k}$  and covariance $\mathbf{P}_{k|k}$ in the following sense

\begin{eqnarray}\label{DM}
\bar \mathbf{x}_{k|k}&=&\sum_{i=1}^{N}{\omega_{i,k|k}\mathcal S_{i,k|k}}\\
\mathbf{P}_{k|k}&=&\sum_{i=1}^{N}{\omega_{i,k|k}(\mathcal S_{i,k|k}-\bar \mathbf{x}_{k|k})(\mathcal S_{i,k|k}-\bar \mathbf{x}_{k|k})^{T}}
\end{eqnarray}

Based on the given sampling, the sample propagation in UKF and CKF could be unified as
\begin{eqnarray}
\label{Eq:Samprop}
\mathcal S_{i,k+1|k}&=&\mathbf f(\mathcal S_{i,k|k})
\end{eqnarray}
And the weight of propagated sample $\mathcal S_{i,k+1|k}$ is as the same as the one of $\mathcal S_{i,k|k}$.

Moreover, the filter process of UKF and CKF could be summarized \cite{2000Gaussianfiltersfornonlinear, 2013High-degreecubatureKalman} as follows.

Time update:
\begin{eqnarray}
\label{Eq:TimeUpdate}
\bar \mathbf{x}_{k+1|k}&=&\sum_{i=1}^{N}{\omega_{i,k|k}\mathcal S_{i,k+1|k}}\\
\mathbf{P}_{k+1|k}&=&\sum_{i=1}^{N}{\omega_{i,k|k}(\mathcal S_{i,k+1|k}-\bar \mathbf{x}_{k+1|k})}\nonumber\\
&&~~~~~\times(\mathcal S_{i,k+1|k}-\bar \mathbf{x}_{k+1|k})^{T}+\mathbf{Q}_{k}
\end{eqnarray}

For further measurement updating, it will need sample random variable $\mathbf{x}_{k+1|k}$ based on its mean $\bar \mathbf{x}_{k+1|k}$ and covariance $\mathbf{P}_{k+1|k}$.  Let  $\mathcal S_{i,k+1|k}^{*}$ and $\omega_{i,k+1|k}$ stand for the samples and corresponding weights respectively.

Measurement update:
\begin{eqnarray}
\bar \mathbf{x}_{k+1|k+1}&=&\bar \mathbf{x}_{k+1|k}+\mathbf{K}_{k}(\mathbf{y}_{k+1}-\bar \mathbf{z}_{k+1|k}) \label{formula:datafusion}\\
\mathbf{P}_{k+1|k+1}&=&\mathbf{P}_{k+1|k}-\mathbf{K}_{k}\mathbf{P}_{zz,k+1|k}\mathbf{K}_{k}^{T}
\end{eqnarray}
where
\begin{eqnarray}
\bar \mathbf z_{k+1|k}&=&\sum_{i=1}^{N}{\omega_{i,k+1|k}\mathcal Z_{i,k+1|k}}\\
\mathcal Z_{i,k+1|k}&=&\mathbf{h}(\mathcal S_{i,k+1|k}^*)\\
\mathbf{P}_{zz,k+1|k}&=&\sum_{i=1}^{N}{\omega_{i,k+1|k}(\mathcal Z_{i,k+1|k}-\bar \mathbf{z}_{k+1|k})}\nonumber\\
&&~~\times(\mathcal Z_{i,k+1|k}-\bar \mathbf{z}_{k+1|k})^{T}+\mathbf{R}_{k+1}\\
\mathbf{P}_{xz,k+1|k}&=&\sum_{i=1}^{N}{\omega_{i,k+1|k}(\mathcal S_{i,k+1|k}^*-\bar \mathbf{x}_{k+1|k})}\nonumber\\
&&~~\times(\mathcal Z_{i,k+1|k}-\bar \mathbf{z}_{k+1|k})^{T}\\
\mathbf{K}_{k}&=&\mathbf{P}_{xz,k+1|k}\mathbf{P}_{zz,k+1|k}^{-1}
\label{Eq:MeasureUpdate}
\end{eqnarray}
where $\mathbf y_{k+1}$s are the measure data.

Note that, there is a significant difference between CKF and UKF on the sampling  $\mathcal S_{i,k+1|k}^*$. The original UKF \cite{1997AnewextensionKalman} just directly takes $\mathcal S_{i,k|k}$  as  $\mathcal S_{i,k+1|k}^*$ with associated weight, when there is no distribution assumption on $\mathbf{x}_{k+1|k}$ \cite{2004Unscentedfiltering}. But the CKF \cite{2009Cubaturekalmanfilters} utilizes cubature rule to resample $\mathbf{x}_{k+1|k}$ based on $\bar{\mathbf{x}}_{k+1|k}$ and $\mathbf{P}_{k+1|k}$, under the Gaussian assumption on $\mathbf{x}_{k+1|k}$. As can be seen in simulation, under the Gaussian assumption, if UKF also resamples $\mathbf{x}_{k+1|k}$ like itself sampling $\mathbf{x}_{k}$ then there is a great improvement on the performance of UKF. So the sampling strategies is very important in such filters. In the following, we review these sampling strategies.

\subsection{Unscented Rule Based Sampling}

In UKF, the  unscented sampling (US) selects  samples (so-called sigma points \cite{2004Unscentedfiltering})  to approximate the probability distribution of a random variable by matching its mean and covariance. As illustrated  in Fig. \ref{fig1:UT},  the samples from the contour are determined by the mean and covariance.  The wildly used second order US \cite{2004Unscentedfiltering}  selects  symmetrical sigma points with $N=2n+1$ as follows:
\begin{eqnarray}
\label{UT1}
\begin{array}{llll}
\mathcal S_{0}&=\bar \mathbf{x}&\omega_{0}&=\kappa/(\kappa+n)\\
\mathcal S_{i}&=\bar \mathbf{x}+(\sqrt{(\kappa+n)\mathbf{P}_{x}})_{i}&\omega_{i}&=1/2(\kappa+n)\\
\mathcal S_{i+n}&=\bar \mathbf{x}-(\sqrt{(n+\kappa)\mathbf{P}_{x}})_{i}&\omega_{i+n}&=1/2(\kappa+n)\\
\end{array}
\end{eqnarray}
where $1\leq i\leq n$; $\kappa\in\mathbb{R}$ is a scale parameter to adjust the distance between the sample and mean point; $(\sqrt{(\kappa+n)\mathbf{P}_{x}})_{i}$ is the $i$-th row or column of the matrix square root of $(\kappa+n)\mathbf{P}_{x}$, which can be computed by Cholesky decomposition. In (\ref{UT1}), $\kappa$ is a freedom to be determined.

Such mean and covariance  matching method is naturally extended to  higher moments matching  \cite{2002Thescaledunscented}, \cite{TheHigherOrderUnscentedFilter}, \cite{1997AConsistentDebiasedMethod}, \cite{ConstrainedNonlinearStateEstimationBasedonTheUKFapproach}. They are distinct from the choices of samples and weights. For example, in  \cite{1997AConsistentDebiasedMethod}, $n+\kappa=3$ should hold if it wants to match the fourth order moment of a univariate Gaussian distribution $\mathbf{x}$. Thus, when $n>3$, $\kappa=3-n<0$, which implies weight $\omega_{0}$ being negative. So the covariance may be indefinite to contribute the instability  of filtering process.

\begin{figure}[!t]
  \centering
  \includegraphics[width=3in]{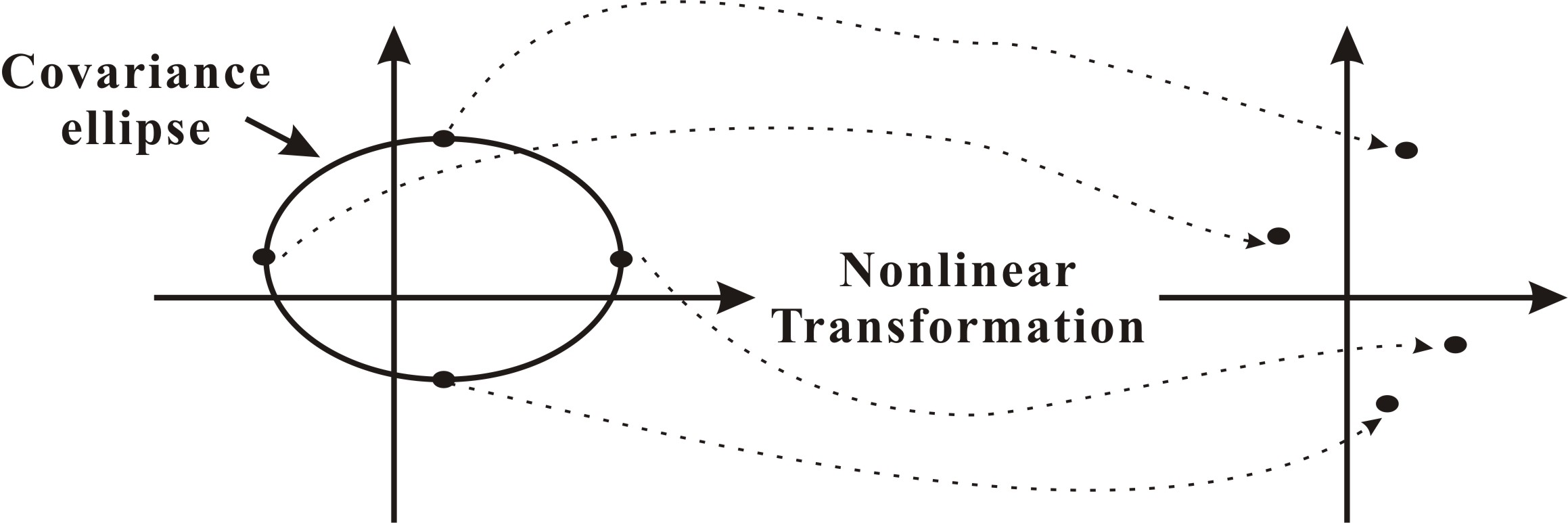}
  \caption{The process of UT}\label{fig1:UT}
\end{figure}

\subsection{Cubature Rule Based Sampling}

In the CKF, to compute posterior distribution, an integral $I(\mathbf f)=\int_{\mathbb R^n} \mathbf{f}(\mathbf{x})\times \mbox{exp}(-\mathbf{x}^{T}\mathbf{x})\mbox{d}\mathbf{x}$ is approximated by samples with associated weights determined by using cubature rule based upon moments matching. Through variable change transformation by setting $\mathbf{x}=r\mathbf{y}$ with $r\geq 0$ and $\mathbf y\in \mathbb R^n$ such that $\mathbf y^T\mathbf y=1$, then  $I(\mathbf f)$ can be rewritten in a spherical-radial coordinate system as
\begin{eqnarray}
\label{SRC}
\begin{array}{lll}
I(\mathbf f) &=& \int_0^{\infty}\int_{U_n}\mathbf f(r\mathbf y)r^{n-1}exp(-r^2)d{\mathbf y} dr
\end{array}
\end{eqnarray}
where $U_n$ is the surface of the sphere specified by $U_n=\{\mathbf y \in\mathbb R^n \mid \mathbf y^T\mathbf y=1\}$.  The {\it spherical-radial cubature rule}  is a combination of  spherical rule, radial rule and cubature rule.  The {\it cubature rule} is about  geometry distribution of samples, which  employs  fully symmetric points and assigns  equal weight to each point.  The {\it spherical rule} is a discretization approach to the integral of form $\int_{U_n}\mathbf f(\mathbf y)d\mathbf y$. And the radial rule is a   discretization approach to the integral of form $\int_{0}^{\infty} f(r)r^{n-1}exp(-r^2)dx$.  Then the samples and the associated weights are computed by solving Gaussian weighted integral equations. For brevity, we call this process cubature sampling (CS).  According to the degrees of $\mathbf f$, the CS is classified into 3-degree \cite{2009Cubaturekalmanfilters} and high-degree  \cite{2013High-degreecubatureKalman}.  The set of samples and weights of 3-degree CS are given by \cite{2009Cubaturekalmanfilters} as follows:
\begin{eqnarray}
\begin{array}{llll}
\mathcal S_{i}&=\bar \mathbf{x}+(\sqrt{n\mathbf{P}_{x}})_{i}&\omega_{i}&=1/2n\\
\mathcal S_{i+n}&=\bar \mathbf{x}-(\sqrt{n\mathbf{P}_{x}})_{i}&\omega_{i+n}&=1/2n\\
\end{array}
\end{eqnarray}
where $1\leq i \leq n$.  For the sampling of high degree CKF, please refer to \cite{2013High-degreecubatureKalman}.

It is clear that 3-degree CS has equal positive weights, which is believed to contribute to the stability comparing with UKF \cite{2009Cubaturekalmanfilters}. Theoretically, high-degree CKF could achieve higher filtering accuracy. However, some weights would be negative, for example, some wights in 5-degree CS are $ \displaystyle\frac{4-n}{2(n+2)^{2}}$ which is negative if $n>4$. The negative weights may lead to unstable calculation process and indefinite result like UKF and halt its operation in CKF,  sabotaging performance presumed as analyzed in \cite{2009Cubaturekalmanfilters}.

\subsection{Importance Sampling}

In the PF, the samples are recursively generated by the so-called  {\it importance sampling} (IS). As a global filter, the PF samples the whole trajectory instead of a single state.  Summarily, the IS \cite{Tokdar2010} refers to a collection of {\it Monte Carlo methods} where a mathematical expectation $\mathbb E_p[\mathbf f (X)]=\int \mathbf f(\mathbf x) p(\mathbf x) d\mathbf x$ with respect to a target distribution  $p(\mathbf x)$  is approximated by a weighted average of random draws from another distribution specified by the weighting function $w(\mathbf x)=\frac{p(\mathbf x)}{q(\mathbf x)}$,  where $p(\mathbf x)$ is the density function of distribution $X$ and $q(\mathbf x)$ is the so-called  {\it importance density} \cite{2009Atutorialonparticlefiltering}. The approximation accuracy by IS highly depends on the choice of $q(\mathbf x)$.  Equivalently, the sampling of $X$ resolves the accuracy of approximation to $\mathbb E_p[\mathbf f (X)]$.

The recipe of IS is to concentrate on the regions where the value is large, and avoid taking samples in regions where the value of the function is negligible \cite{2000Importancesamplingillustrative}. It means that the amount of samples of a region should be proportional to the value the region has. If we think that a sample represents its neighbour region, then big value region means a big valued weight for its represented samples. We may get some intuition from the following simple example. Let $X$ be a random variable and its density distribution be the triangular function $f(x)=0.5(x-a)$ with $x\in [a, a+2]$ for some positive number $a$, which could be plotted  like the left graph  of Fig. \ref{fig:CUMUFUN}.
\begin{figure}[!t]
	\centering
	\includegraphics[width=3in]{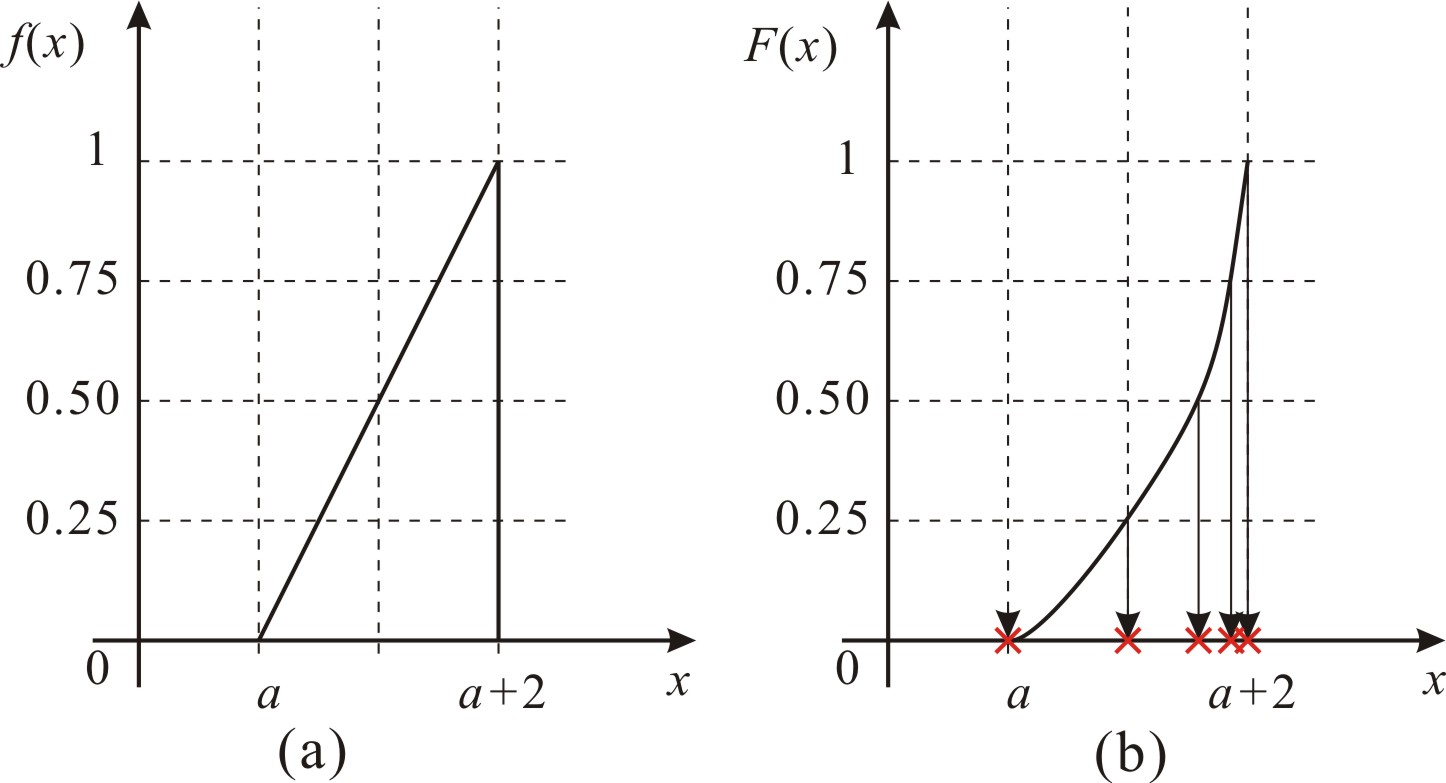}
	\caption{(a) Density distribution. (b) Its cumulative distribution.}\label{fig:CUMUFUN}
\end{figure}
Accordingly, its cumulative distribution is $F(x)=0.25(x-a)^{2}$ on the right of Fig. \ref{fig:CUMUFUN}. Now we would think $F(x)$ as $q(\mathbf x)$ (but not exactly). Then an IS could  carry out as follows:
\begin{enumerate}
	\item  Generate  random numbers $\zeta_{i} \in [0,1], i=1,2,\cdots, N$.
	\item Then set $F(x_{i})=\zeta_{i}$, and solve the $x_{i}= F^{-1}(\zeta_{i})$,
\end{enumerate}
where the $F^{-1}$ stands for the inverse function of $F$. As $F(x)$ is a cumulative distribution function, the intervals with the same  $F$-difference between its end points have the equal importance. So all the samples  from such intervals have the equal importance. That is, a uniform $F$-value distribution would give an importance sampling. For example, in Fig. \ref{fig:CUMUFUN},  the values $\zeta=\{0, 0.25, 0.5, 0.75, 1\}$ give an importance sampling $F^{-1}(\zeta)$'s  associated with weights $\omega_i$ as follows.
\begin{eqnarray}
\label{ISEXAMPLE}
\begin{array}{lllll}
\mathcal X_{1}&=F^{-1}(0)&= a, &  \omega_{1}&=0 \\
\mathcal X_{2}&=F^{-1}(\frac{1}{4})&= a+1, &  \omega_{2}&=\frac{1}{\sqrt{2}+\sqrt{3}+3} \\
\mathcal X_{3}&=F^{-1}(\frac{1}{2})&= a+\sqrt{2}, & \omega_{3}&=\frac{\sqrt{2}}{\sqrt{2}+\sqrt{3}+3} \\
\mathcal X_{4}&=F^{-1}(\frac{3}{4})&= a+\sqrt{3}, &  \omega_{4}&=\frac{\sqrt{3}}{\sqrt{2}+\sqrt{3}+3} \\
\mathcal X_{5}&=F^{-1}(1)&= a+2, & \omega_{5}&=\frac{2}{\sqrt{2}+\sqrt{3}+3}
\end{array}
\end{eqnarray}
These sample points are labeled by red crosses on $x$-coordinate in Fig. \ref{fig:CUMUFUN}. It is clear that the number of sample points in each interval is proportional to the $f$-value at the center of that interval.

The IS-based PF performs quite well in three-dimensional state space. However, it suffers from the curse of dimensionality, which makes the particle representation too sparse to be a meaningful representation of the posterior distribution in higher dimension cases \cite{Gustafsson10particlefilter}. In practice, the performance degrades quickly with the state dimension.

\section{Uniformly Geometric Unscented Filter}\label{GUF}

Each of the US, CS and IS is some kind of approximation for  $\int_{\mathbb R^n} \mathbf f(\mathbf x) p(\mathbf x) d\mathbf x$ with different accuracy, efficiency  and reliability. The US and CS are of good efficiency, but short on accuracy or reliability in case of acute nonlinearity and high dimensions. The IS showed good accuracy, but its high computational complexity hampers the application for higher dimensions.  Summarily, the common challenge of all these methods is how to develop them for high dimensions. This section presents a scalable sampling scheme to tackle the challenge.

The basic idea of our novel sampling strategy is to reduce  higher dimensional sampling to one dimensional case through an {\it importance function} (IF). The IF indicates the value of points. Then we make judicious choice of samples according to IF. Given a random variable $X$, it is widely accepted that the mean $\bar{\mathbf x}$ and covariance $cov(\mathbf x)$ has the highest importance. Indeed, in most cases, if these two and distribution are known then the density function $p(\mathbf x)$ could be completely formulated. In fact,  US and CS  make full use of this feature in their sampling. In this line of thinking, we say a positive value function $i: \mathbb R^n \mapsto \mathbb R^+$ an {\it importance function} for an $n$-dimension random variable with mean $\bar{\mathbf x}$, density distribution $p(\mathbf x)$ and probability distribution $P(X)$, if it satisfies
\begin{enumerate}
\item $i(\bar{\mathbf x})\geq i(\mathbf x) >0$, for all $\mathbf x\in \mathbb R^n$,
\item $i$ is upper semi-continuous,
\item $S_\geq(d):=\{\mathbf x\in \mathbb R^n \mid i(\mathbf x)\geq d\}$ is compact and connected for any $d\in \mathbb R^+$ with $d> i_0$,
\item for $d\in \mathbb R^+$, $p(\mathbf x)=p(\mathbf y)$ for any $\mathbf{x, y}\in S_=(d)$,  where $S_=(d):=\{\mathbf x\in\mathbb R^n \mid i(\mathbf x)=d\}$,
\item $Y$ is a uniform distribution,
\end{enumerate}
where $i_0=\inf\limits_{\mathbf x\in \mathbb R^n} i(\mathbf x)$ and $Y$ is the distribution on $[i_0, i(\bar\mathbf x)]$ derived from probability distribution $P(X\in S_\geq(d))$ for $d\in [i_0, i(\bar\mathbf x)]$.
Term 1) emphasizes the super importance of mean value. The continuity in  2) is a smoothing requirement for the IF. For a finite approximation, it imposes the compactness in 3). The condition 4) depicts an equal density distribution for equal importance points. A uniform distribution of $d$ in 5) restricts the concentration on the IF such that we can unbiasly consider all importance references from IF. Note that the IF is not necessarily injective, i.e., $|S_=(d)|\neq 1$, where $|S|$ denotes the amount of elements in  $S$. This means an equal importance of the different points in $S_=(d)$. On the other side, $S_=(d)$ and $S_\geq(d)$ roughly specify certain region value specified by
\begin{eqnarray} \label{eq:IR}
i_R(D)=d_2-d_1, \quad D=\{\mathbf x\in\mathbb R^n\mid d_1\leq  i(\mathbf x)\leq d_2\}
\end{eqnarray}
Note that, the function $i_R(\cdot)$ is a partial function over its domain $\mathcal P(\mathbb R^n)$, where $\mathcal P(\mathbb R^n)$ is the power set of $\mathbb R^n$.

 So the IF is critical to our sampling. In what follows, we discuss an IF based on the mean value, covariance matrix and probability distribution.

\subsection{An Importance Function W.R.T. Gaussian Distribution}
Given a Gaussian random variable $\mathbf x$ and its probability distribution $P(X)$ specified by a probability density function  $p(\mathbf x)$. Let $\bar{\mathbf x}$ and $cov(X)=\mathbf P_X$ be the mean value and covariance matrix. Then, the density function of $\mathbf x$  is
\begin{eqnarray} \label{Eq:GDF}
p(\mathbf{x})=\lambda \exp\left(-\frac{1}{2}(\mathbf{x}-\bar \mathbf{x})^{T}\mathbf{P}_{X}^{-1}(\mathbf{x}-\bar \mathbf{x})\right),~\mathbf{x}\in \mathbb{R}^{n}
\end{eqnarray}
where $\displaystyle\lambda=\frac{1}{{(2\pi)}^{n/2}}\frac{1}{|\mathbf{P}_{X}|^{1/2}}$.
To construct an importance function for $\mathbf x$ , we set a real value function $L: \mathbb R^n\mapsto \mathbb R$ as
\begin{eqnarray}
\label{LF}
L(\mathbf{x})=(\mathbf x-\bar\mathbf x)^T\mathbf{P}_{X}^{-1}(\mathbf x-\bar\mathbf x)
\end{eqnarray}
where  $\mathbf{P}_{X}^{-1}$ is the inverse of $\mathbf{P}_{X}$. Note that, $L$ is a positive definite function, since $\mathbf P_X$ is a positive definite matrix. We using $L$ define radial region $D_r$ as
\begin{eqnarray} \label{RR}
D_r=\{\mathbf x\in \mathbb R^n \mid  L(\mathbf{x})\geq r\}
\end{eqnarray}
and then take the integral function $R_C: \mathbb R^* \mapsto [0,1]$
\begin{eqnarray}
\label{RC}
R_C(r)=\int_{\mathbf x\in D_r} p(\mathbf x)d\mathbf x
\end{eqnarray}
where $\mathbb R^*$ is the set of nonnegative numbers.  It is obvious that $L(\bar\mathbf x)=0$ and $R_C(0)=1$.

Now we define an importance function $i:\mathbb R^n\mapsto \mathbb R^+$  as
\begin{eqnarray}
\label{IFRC}
i(\mathbf x) & = & R_C(L(\mathbf x)) \\
            & = &\frac{1}{(2\pi)^{n/2}} \int_{\mathbf y:~ \mathbf y^T\mathbf y\geq L(\mathbf x)}\exp{(-\frac{1}{2}\mathbf y^T\mathbf y)d\mathbf y} \label{eq:distr}
\end{eqnarray}
First, $i(\mathbf x)>0$ and $i(\bar\mathbf x)=1\geq i(\mathbf x)$ for all $\mathbf x\in \mathbb R^n$. It is evident that $i(\cdot)$ is upper semi-continuous by the continuity of $R_C(\cdot)$ and $L(\cdot)$. From (\ref{RC}),  $i_0=0$, together with (\ref{IFRC}), it follows for any $d>i_0$
\begin{eqnarray}\label{CCC}
S_\geq(d)=\{\mathbf x\mid i(\mathbf x)\geq d\}=\{\mathbf x\mid L(\mathbf x)\leq r_d\}
\end{eqnarray}
where $r_d$ is the real number such that $R_C(r_d)=d$. So $S_\geq(d)$ is compact and connected for all $d>i_0$. Furthermore, it is not hard to show any $\mathbf x,\mathbf y\in S_=(d)$ implies $p(\mathbf x)=p(\mathbf y)$ herein. Let $L_=(r):=\{\mathbf x\mid L(\mathbf x)=r\}$ for the real numbers $r>L(\bar\mathbf x)$, then it is actually  $S_=(d)=L_=(r_d)$ for any $d>i_0$. At last, for the derived distribution $Y$ from $P(X\in S_\geq(d))$, we have
\begin{eqnarray}\label{URD}
P(Y\geq d)=P(X\in S_\geq(d))=P(D_{r_d})=R_C(r_d)=d
\end{eqnarray}
for $d\in[0, 1]$ and so it has a uniform distribution. Therefore, the function $i(\cdot)$ defined by (\ref{IFRC}) is an IF.

This IF is induced by the probability distribution function, which mainly concentrates on the characters of the probability distribution. In subsequent research, we are going to further investigate the construction of IF that is related to the integrand  $\mathbf f(\cdot)$ in $\int_{\mathbb R^n}{\mathbf f(\mathbf x)}p(\mathbf x)d\mathbf x$. Moreover, it might be useful to study the IF that considers the characters of both integrands $p(\cdot)$ and  $\mathbf f(\cdot)$.

\subsection{Uniformly Geometric Unscented Sampling}

Based on previous importance function $i(\mathbf x)$, our sampling strategy runs as follows.
\begin{enumerate}
	\item[i)] First, it generates uniformly distributed random numbers $d_k\in [i_0, i(\bar\mathbf x)]$ with $k=1,2, \dots, N$ for some integer $N$.
	\item[ii)] Then we pick basic samples $\mathcal X^*_{kj}$ such that $i(\mathcal X^*_{kj})=d_k$, with $1\leq j\leq N_k$, for some integer $N_k$. For fixed $k$, all $\mathcal X^*_{kj}$s share the same importance value $d_k$ and so should have the same weight $\omega^*_{kj}$.
	\item[iii)] At last, we normalize the weights $\omega^*_{kj}$ and match the moments of $\mathbf x$ through adjusting the basic samples to obtain the final samples. Eventually, they jointly make a density approximation to the random variable $X$.
\end{enumerate}

First of all, $S_=(d_k)$s are disjoint for any sequence $d_1<\cdots<d_N$ and divide the spaces $\mathbb R^n$ into at most $2N+1$ many disjoint connected parts.  It is natural to select the points from $S_=(d_k)$ to represent the region $R_k:=S_\geq(d_k)\cap S_\geq(d_{k-1})$ for $1<k\leq N$ and $R_1:=\mathbf R^n - S_\geq(d_1)$ for $k=1$.
There are still three challenges to carry out the sampling scheme.
\begin{enumerate}
	\item[CH1] How many samples should be taken for a given $d_k$?
	\item[CH2] How to choose basic samples $\mathcal X^*_{kj}$ from the set  $S_=(d_k)$?
	\item[CH3] What should be the proper weight distribution  $\omega^*_{kj}$?
\end{enumerate}

For CH1, we  employ the idea of IS to decide the amount of samples associated with $d_k$ in what follows. As the points from $S_=(d_k)$ could represent the region $R_k$,  the importance value of $R_k$ may be chosen as $i_R(R_k):=d_k-d_{k-1}$, the amount of samples from $S_=(d_k)$ should be proportional to $i_R(R_k)$. When the density function $p(\mathbf x)$ is known, we  consider an alternative option that takes the samples $\{\mathcal X_{kj}\}$ with $N_k$ proportional to the density function value $p(\mathcal X_{kj})$, since the density value reflects certain importance of samples with respect to their surrounding regions. In PF, the samples are randomly picked from  $R_k$ of size proportional to $i_R(R_k)$.

For CH2, we extend the idea of cubature rule by sampling symmetrically and evenly distributed in the set $S_=(d_k)$, since all points in $S_=(d_k)$ have the same importance value. This can be achieved because of the symmetry of $S_=(d_k)$ to the mean. Moreover, the samples here require more symmetry on the generators than CKF does, but it does not globally require the equal weights. To this end, we come up with a notion of {\it uniformly geometric distribution} (UGD). First we consider the UGD on the sphere $U_n:=\{\mathbf x\in \mathbb R^n\mid \mathbf x^T \mathbf x=1\}$.  A finite sample set $S\subset U_n$ is called a UGD if it satisfies the following conditions:
\begin{enumerate}
	\item[a)] Each $\mathbf x\in  S$ implies $\tau(\mathbf x)\in S$, where $\tau$ is an operation on coordinates of $\mathbf x$ which implements the permutation and/or sign changes of the coordinates.
	\item[b)] There is a constant $d^*$ such that for all $\mathbf x\in  S$, $\inf_{\mathbf y\in  S} \|\mathbf x-\mathbf y\|=d^*$, where $\|\cdot\|$ is the Euclidean norm.
	
\end{enumerate}
The UGD, especially term b), presents a globally even spatial distribution of samples. In the next section, under the Gaussian density assumption, we will elaborate how to apply the spheres' UGD to general sets like $S_=(d_k)$ which are not necessarily $U_n$ anymore.

For CH3, it is natural to take $\omega^*_{kj}=p(\mathcal X^*_{kj})$ if the density distribution $p(\mathbf x)$ of $\mathbf x$ is known. This paper follows such rule under the density distribution assumption. When the density distribution $p(\mathbf x)$ is unknown, we suggest to consider the importance value  $i(\mathcal X^*_{kj})$ as the basic weight $\omega^*_{kj}$ of sample $\mathcal X^*_{kj}$. Similar to UKF, through moments matching, we compute the
normalized weights and adjust the samples $\mathcal X^*_{kj}$ to approximate the density distribution of random variable $X$. Note that, the final importance samples $\mathcal X_{kj}$  are usually different from  $\mathcal X^*_{kj}$ after the moments matching adjustment.

Recall the sampling process, once the numbers $d_k$ are chosen, everything else is deterministic. Nevertheless, the choices of $d_k$ are not fully random since there is still a uniform distribution requirement. Anyway,  $d_k$ can be generated by using the Monte Carlo method for one-dimensional space case. On the other side, the basic sample set  $S\subset U_n$ can be deterministically chosen with highly spatial uniform distribution. So, this sampling is called {\it geometric unscented sampling} (GUS), which is a   semi-deterministic sampling strategy.

\begin{figure*}[!t]
  \centering
  \includegraphics[height=2in]{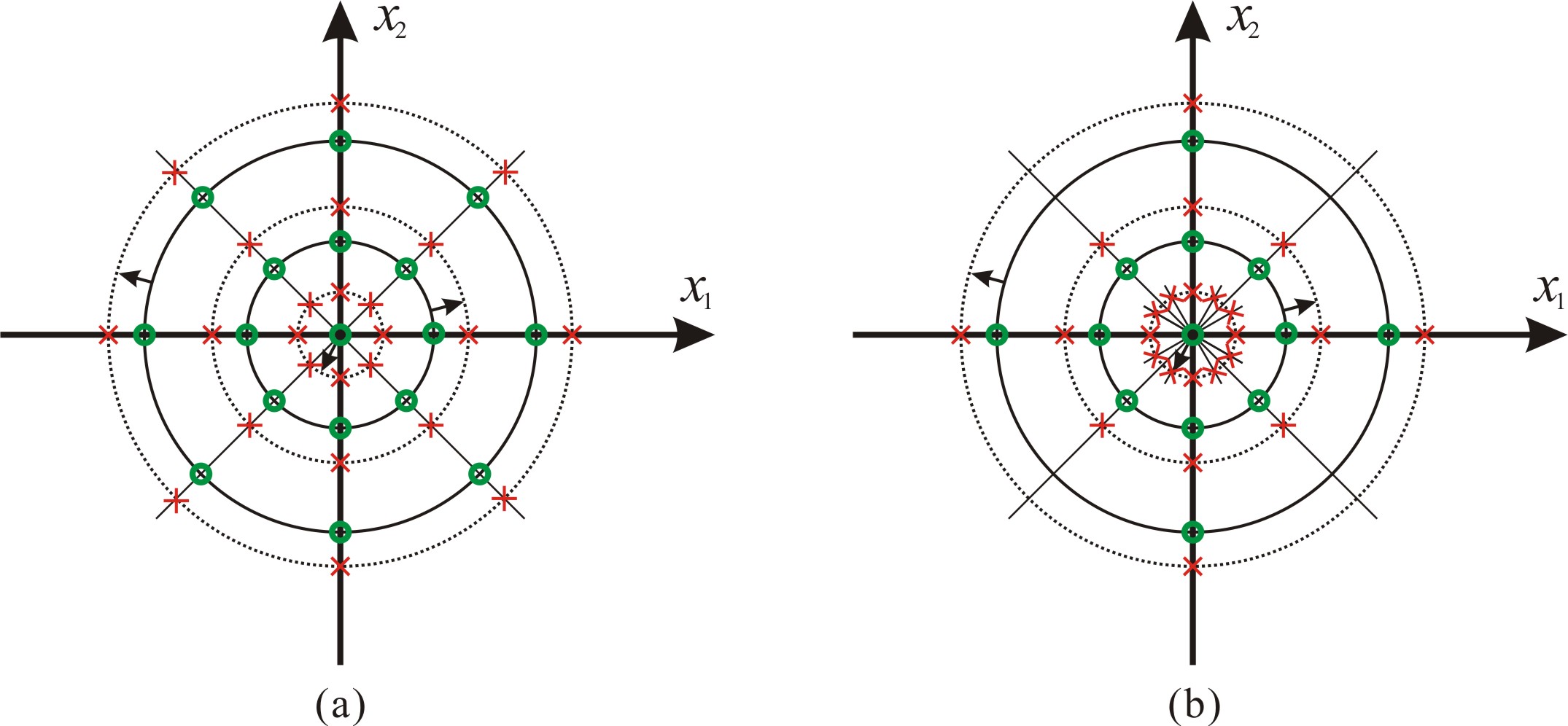}
  \caption{Umbrella Form Importance Sampling of a two-variate normal distribution with $N=3$. (a) Importance samples chosen based upon equation (\ref{NumR}). (b) Importance samples were chosen based upon the equation (\ref{NumDR}).}\label{fig:MIS1}
\end{figure*}

{\bf Example.} Let $X$ be an $n$-variate normal distribution with mean vector $\mathbf 0$ and $n$-dimensional unit covariance matrix $U$. The density function of $X$ is $p(\mathbf x)=\frac{1}{\sqrt{(2\pi)^n}}\exp (-\frac{1}{2}\mathbf x^T\mathbf x)$. Let's illustrate the GUS by taking $n=2$. Without loss of generality, we may take $d_k=\frac{k}{N}$ in $[0,1]$ for $1\leq k\leq N$. And so $i_R(R_k)=\frac{1}{N}$ for all $k$. Using  (\ref{GR}) to compute $r_k$
\begin{eqnarray}
\label{GR}
\int_{\mathbf x\in D_k} \frac{1}{2\pi}\exp (-\frac{1}{2}\mathbf x^T\mathbf x) d\mathbf x=d_k
\end{eqnarray}
where $D_k=\{\mathbf x\in \mathbb R^n \mid  L(\mathbf{x})\geq r_k\}$ and $L(\mathbf x)=\mathbf x^T \mathbf x$. Then we select basic samples $\mathcal X^*_{k,j}$ evenly from the circles   $S_=(r_k)=\{\mathbf x\in \mathbb R^2 \mid \mathbf x^T \mathbf x=r_k\}$. As the density function $\frac{1}{2\pi}\exp (-\frac{1}{2}\mathbf x^T\mathbf x)$ is known, we take $\omega^*_{k,j}=\frac{1}{2\pi}\exp(-\frac{1}{2}r_k)$. The numbers of different importance samples may be chosen such that

\begin{eqnarray}
\label{NumR}
N_1:\cdots: N_N\propto 1:\cdots : 1
\end{eqnarray}
or
\begin{eqnarray}
\label{NumDR}
N_1:\cdots: N_N\propto \exp(-\frac{1}{2}r_1):\cdots : \exp(-\frac{1}{2}r_N)
\end{eqnarray}
where $N_k$ is the number of samples with importance values $\leq d_k$. At last, we normalize $\omega^*_{k,j}$ by
\begin{eqnarray}
\label{NWF}
\omega_{k,j}=\frac{\omega^*_{k,j}}{\sum_{k,j}{\omega^*_{k,j}}}.
\end{eqnarray}
Note that, the final value of weights depends on both of the $\omega^*_{k,j}$ and the amount of samples.

For a simple illustration, we take $N=3$, $d_1=1/3$, $d_2=2/3$, and $d_3=1$. Then $i_R(R_k)=1/3$ for all $k=1, 2, 3$.  Using (\ref{GR}) computes $r_k$, it obtains $r_1=2.1972$, $r_2=0.8109$, $r_3=0$.  We could express such information by a diagram on two dimension space like the  Fig. \ref{fig:MIS1}. Wherein,  the solid circles stand for the sets $S_=(r_1)$, $S_=(r_2)$ and $S_=(r_3)$ where basic samples $\mathcal X^*_{k,j}$ locate with $\omega^*_{1,j}=0.0531$, $\omega^*_{2,j}=0.1061$ and $\omega^*_{3,j}=0.1592$. The numbers of basic samples for different importance value $d_k$ can be chosen based upon (\ref{NumR}) or (\ref{NumDR}).

According to (\ref{NumR}), the numbers $N_k$ of basic samples should obey to $N_1: N_2:N_3\propto 1:1:1$.  On the other side,  the basic samples should be symmetric.  In this example, we take $N_k=8$ which allows symmetric samples $\sqrt{r_k}\mathcal S$ for each importance value $r_k$ where the $\mathcal S$ is consist of  total permutations and/or sign changes of the $[1, 0]^T$ and $[\frac{\sqrt{2}}{2}, \frac{\sqrt{2}}{2}]^T$.  Such basic samples are described by the green circles in Fig. \ref{fig:MIS1}(a). Note that,  as $r_3=0$,  the point $(0,0)$ is a special sample which is a collapse of all  $\sqrt{r_3}\mathcal S$ and could be seen as an eight-fold overlap. In general, we take the importance samples of form $\sqrt{r_k+\beta}\mathcal S$ to give a freedom for moments matching. As for CH3, we assign the importance samples with normalized weights $ \omega_{1,j}= 0.0208$, $\omega_{2,j}=0.0417$ and $\omega_{3,j}=0.0625 $. By moments matching adjustment, we have $\beta= 1.3635$. Eventually, we obtain the importance samples described by the red crosses in Fig. \ref{fig:MIS1}(a).

In case of the proportion  (\ref{NumDR}), for the sake of symmetry we take $N_1=4$, $N_2=8$ and $N_3=12$, which are roughly $4:8:12\propto \exp(-\frac{1}{2}r_1):\exp(-\frac{1}{2}r_2):\exp(-\frac{1}{2}r_3)$.  For each $r_k$, its basic samples sets are $\sqrt{r_1}\mathcal S_1$, $\sqrt{r_2}\mathcal S_2$ and $\sqrt{r_3}\mathcal S_3$, respectively, which are described by the green circles in Fig. \ref{fig:MIS1}(b). Here, the  $\mathcal S_1$ is consist of  total permutations and/or sign changes of the $[1, 0]^T$. So is the $\mathcal S_2$ of the $[1, 0]^T$ and $[\frac{\sqrt{2}}{2}, \frac{\sqrt{2}}{2}]^T$, and so is the $\mathcal S_3$ of the $[1, 0]^T$ and $[\frac{\sqrt{3}}{2}, \frac{1}{2}]^T$.  Similar to the case (\ref{NumR}), the importance samples are of form $\sqrt{r_k+\beta}\mathcal S$ and the normalized the weights become  $ \omega_{1,j}= 0.0179$, $ \omega_{2,j}= 0.0357$, $ \omega_{3,j}=0.0536 $. Through moments matching, it obtains $\beta=1.6114$. Eventually, we obtain the corresponding importance samples described by the red crosses in Fig. \ref{fig:MIS1}(b).

\subsection{Nonlinear Filters Based On GUS}
Now we can utilize GUS to conduct a filter called {\it geometric unscented Filter}  (GUF). Roughly, the GUF  share a similar filtering framework to UKF and CKF, through the formulas (\ref{DM})-(\ref{Eq:MeasureUpdate}). Similar to CKF,  the GUF takes advantage of resampling $\mathbf x_{k+1\mid k}$ under Gaussian assumption; otherwise, it directly uses the transformed samples like UKF. The major difference is that the  GUF employs the sampling strategy GUS to compute the samples  $\mathcal S_{i,k|k}$, $\mathcal S^*_{i,k+|k}$ and the corresponding weights. By the GUS, it is clear that the weights are always positive, which is ensured by a process like the sampling in PF. The basic samples of GUS are selected by an extended method of CS and US. The final samples are computed by a moment matching rule like the US. Thus, GUF is developed out from PF, UKF and CKF.

Note that, due to the limitation of the filtering framework of GUF, the accuracy loss by such filtering framework cannot be avoided even more advanced sampling method is used. For an arbitrary accuracy estimate, we may study the nonlinear filter which adopts the filtering scheme of PF and sampling strategy GUS. However, this is not the goal of this article,  which will be explored in another work.

The GUS conducts a certain simple random resampling in one-dimensional space, but the IS conducts more complex random resampling in higher-dimensional space. In a certain sense, we reduce the complex random resampling of IS to a simpler one using GUS. Comparing the CS of CKF with the US of UKF, the US relies on special samples with the symmetric property.
Note that the CS could be a special case of GUS with $N=1$, some appropriate choice of $d$ and moments matching.  Similarly, the US could also be a special case of GUS with special values of $d$ and  $N=1$.  However, the GUS considers the contributions not only from moments by moments matching but also from probability distribution through the importance function. More importantly, this allows the GUS to employ arbitrarily many samples to approximate a given probability distribution as accuracy as desired at a reasonable cost. As for the weights, the GUS ensures all the weights to be in  $[0,1]$ and to take arbitrarily many different values, unlike the high order US and CS that allow negative weights and take only two different values.

\section{Uniformly Geometric Distribution}\label{sec:UGD}

The previous section gave a framework of GUS and left the discussion about UGD to this section. We present here a detailed GUS of Gaussian random variables called  {\it geometric unscented Gaussian sampling} (GMCGS).

Let $X$ be an $n$-dimensional Gaussian distribution with mean $\bar{\mathbf x}$ and covariance $\mathbf{P}_{X}$, then the density function is (\ref{Eq:GDF}).
We utilize the IF defined by (\ref{IFRC}) since the density function is known here. In this case, the importance values are in the interval $[0,1]$. As a special case of the GUS method, the GMCGS first generates a  uniformly distributed random numbers $d_k\in [0,1]$ ordered by increasing  $k$, where $1\leq k\leq N$ for some integer $N$.

\subsection{Computing The Basic Samples' Parameters}
In what follows, we come to a crucial step of  GMCGS for  generating the basic samples $\mathcal X^*_{kj}$ such that $i(\mathcal X^*_{kj})=d_k$. To this end, according to (\ref{RC}) and (\ref{IFRC}),  we first compute $r_k$ such that $R_C(r_k)=d_k$ and then choose  $\mathcal X^*_{kj}$ from $L_k:=\{\mathbf x\mid L(\mathbf x)=r_k\}$ by employing the general UGD in the next subsection. Here, $r_k$ is computed through the following formula
\begin{eqnarray}
\label{EQ:}
R_C(r_k)&=&\int_{\mathbf x\in D_{r_k}}{\lambda \exp\left(-\frac{1}{2}(\mathbf{x}-\bar \mathbf{x})^{T}\mathbf{P}_{X}^{-1}(\mathbf{x}-\bar \mathbf{x})\right)d\mathbf x} \nonumber \\
&=&\int_{\mathbf{y}^{T}\mathbf{y}\geq r_k}\frac{1}{{(2\pi)}^{n/2}}\exp\left(-\frac{1}{2}\mathbf{y}^{T}\mathbf{y}\right) \mbox{d}\mathbf{y} \label{radial} \\
&=&d_k
\end{eqnarray}
Let $R_S(r_k)$ denote the right of (\ref{radial}). The integrand of $R_S(r_k)$ is actually the density function of standard Gaussian distribution. Let $\mathbf y=r\mathbf s$ with $r\geq 0$ and $\mathbf s\in \mathbb R^n$  such that $\mathbf s^T \mathbf s=1$, then $\mathbf y^T \mathbf y=r^2$ and hence
\begin{eqnarray}
\label{Eq:ACCSphereRadiu}
R_S(r_k)&=&\frac{1}{{(2\pi)}^{n/2}}\int_{r^{2}\geq r_k}r^{n-1}\exp(-\frac{1}{2}r^{2})\mbox{d}r\int_{\mathbf{U}_{n}}d\sigma(\mathcal S)\nonumber\\
&=&\gamma\int_{r^{2}\geq r_k}r^{n-1}\exp(-\frac{1}{2}r^{2})\mbox{d}r
\end{eqnarray}
where $ \displaystyle\gamma=\frac{1}{2^{n/2-1}}\frac{1}{\Gamma(n/2)}$ with the Gamma function $\Gamma(\cdot)$,  $\mathcal S=(s_{1},\cdots,s_{n})^{T}$ and $\sigma(\cdot)$ is the spherical surface measure or the area element on $U_{n}$.
For the fixed dimension $n$, the original problem is transformed into computing $r_k$ such that
\begin{eqnarray}\label{eq:RPC}
\int_{r^{2}\geq r_k}r^{n-1}\exp(-\frac{1}{2}r^{2})\mbox{d}r =\frac{d_k}{\gamma}
\end{eqnarray}
This is a one-dimensional integral problem. It can be quickly solved by some numerical method.

\subsection{Generic UGD Sampling W.R.T. Gaussian Distribution}\label{sec:UGDforGD}
Now we are going to extend the UGD sampling from $U_n$ to general set like $L_k$ for generating basic samples, under the Gaussian assumption.
The basic idea  is to transform the UGD sampling sets $S$ on $U_n$ into samples on $L_k$ by using the mean $\bar{\mathbf x}$ and covariance $\mathbf{P}_{X}$. Such UGD set $S$ is also called {\it reference sampling}. To this end, we compute the Cholesky matrix decomposition $\sqrt{\mathbf{P}_{X}}$ of $\mathbf{P}_{X}$ such that $\mathbf{P}_{X}=\sqrt{\mathbf{P}_{X}}\sqrt{\mathbf{P}_{X}}^T$. Based on a UGD sampling set $S$ of $U_n$ by the method in the last section, we select the samples:
\begin{eqnarray}
\label{UGDX}
\mathcal X^*_{kj}=\bar{\mathbf x}+ \sqrt{r_k} \sqrt{\mathbf{P}_{X}}\mathcal S_j, \mathcal{S}_j\in S
\end{eqnarray}
It is easy to verify that $L(\mathcal X^*_{kj})=r_k$. That is, $\mathcal X^*_{kj}$ are samples on $L_k$. This is a set of symmetric points with respect to the mean $\bar{\mathbf x}$.
However, they are not necessarily closed under permutations.

Anyway, all samples with the same $r_k$ and different $\mathcal S_j$ have the same importance value $d_k$,   shown by the formulas (\ref{LF})-(\ref{IFRC}) and  (\ref{UGDX}). As for the samples in different $L_k$s, they would possess their weights. Moreover, the number of samples in each  $L_k$ would roughly follow some prior proportion rules as before. This can be realized through a series of different UGDs $S_1, S_2, \cdots$ on $U_n$ such that their samples' amounts can form the required proportion.

\subsection{Normalization of Weights and Moments Matching}\label{sec:NWandMM}
As the density function $p(\mathbf x)$ is known, the basic weight $\omega^*_{kj}$ of sample $\mathcal X^*_{kj}$ of form (\ref{UGDX}) is $\lambda \mathbf\exp(-\frac{r_k}{2})$ computed by (\ref{Eq:GDF}). However, the sampling $(\mathcal X^*_{kj}, \omega_{kj})$ itself is not a proper approximation to $X$, since it has a different covariance from $X$. Even worse, the summation of $\omega^*_{kj}$ is  not  unit in general.

A reasonable sampling with weighted samples should capture the statistics of a random variable. For a Gaussian random, the first two moments  present all information. It is natural to consider the covariance matching. To this end, we adjust the samples $\mathcal X^*_{kj}$ by some uniform stretch on them as follows
\begin{eqnarray}
\label{UIS}
\mathcal X_{kj}&=&\bar{\mathbf x}+ \sqrt{r_k+\beta} \sqrt{\mathbf{P}_{X}}\mathcal S_{j},  \quad \mathcal{S}_{j}\in S_k \label{Eq:samples}
\end{eqnarray}
where $S_k\subset U_n$ are UGD samplings set such that
\begin{eqnarray}
\label{NUIS}
|S_1|:|S_2|:\cdots:|S_N|\propto N_1:N_2:\cdots:N_N
\end{eqnarray}
Accordingly, we compute the samples'  weighting values by
\begin{eqnarray}
\label{Eq:DUIS}
p(\mathcal X_{kj})&=&\lambda \exp(-\frac{1}{2}(\mathcal X_{kj}-\bar{\mathbf x})\mathbf P^{-1}_X(\mathcal X_{kj}-\bar{\mathbf x})^T) \nonumber \\
&=&\lambda\exp(-\frac{1}{2}(r_k+\beta))
\end{eqnarray}
And then we normalize these weights by
\begin{eqnarray}
\label{Eq:NW}
\omega_{kj}&=& \frac{p(\mathcal X_{kj})}{\sum_{k\leq N, j\leq N_j}p(\mathcal X_{kj})}  \nonumber \\
&=& \frac{\exp\{-\frac{1}{2}r_k\}}{\sum_{k\leq N}\exp\{-\frac{1}{2}r_k\}\sum_{j\leq N_k}j}
\end{eqnarray}
Here, the formula confirms that for a fixed $k$, all samples $\mathcal X_{kj}$ for different $j$ have the same weight. Let $w_k$ denote the same value of all $\omega_{kj}$ for a fixed $k$.
By matching the mean and covariance, we have the following equations
\begin{eqnarray}
\bar \mathbf{x}&=&\sum_{k\leq N}\sum_{j\leq N_k}\omega_{kj}\mathcal X_{kj} \nonumber \\
&=& \sum_{k\leq N}\sum_{j\leq N_k}\omega_{kj} \bar\mathbf x \\
\mathbf{P}_{X}&=&\sum_{k\leq N}\sum_{j\leq N_k}\omega_{kj}(\mathcal X_{kj}-\bar \mathbf{x})(\mathcal X_{kj}-\bar \mathbf{x})^{T} \nonumber \\
&=& \sum_{k\leq N}w_k(r_k+\beta)\sqrt{\mathbf P_X}\mathbf{M}_k\sqrt{\mathbf P_X}^{T}
\end{eqnarray}
where  $\mathbf{M}_k=\sum_{j\leq N_k}\mathcal S_{j}\mathcal S_{j}^{T}$ for $\mathcal S_{j}\in S_k$. For each $k$,  let  $B_k$ be a standard basis of $S_k$.  Then by the symmetry of $S_k$ on $U_n$ and its closeness under permutations,  we get
\begin{eqnarray}\mathbf{M}_k=\sum\limits_{\mathcal B\in B_k} H_n(\mathcal B) \mathbf{E}_n
\end{eqnarray}
where $H_n(\mathcal B)$ is real number which can be effectively computed as in appendix \ref{appendix:proofs}.
Let $c_k=\sum\limits_{\mathcal B\in B_k} H_n(\mathcal B)$.  These induce two equations
\begin{eqnarray}
1 &=& \sum_{k\leq N, j\leq N_k}\omega_{kj} \label{eq:Wunit}\\
1&=& \sum_{k\leq N}w_k(r_k+\beta)c_k \label{WCov}
\end{eqnarray}
The equation (\ref{eq:Wunit}) is obviously true by (\ref{Eq:NW}). Once we solve (\ref{WCov}), the GUS is accomplished. That is an easy job, since it is a linear function of $\beta$. In fact,
\begin{eqnarray}\label{eq:SC}
\beta &=&\frac{1- \sum_{k\leq N}w_k r_k c_k}{\sum_{k\leq N}w_k c_k}
\end{eqnarray}

 By this value, the final samples $\mathcal X_{kj}$ defined by (\ref{Eq:samples}) match first two the moments of $X$.

\subsection{Theoretical Analysis of GUS}\label{TAGUS}
Let $\delta:=\max\limits_{1\leq i <k}\{d_{i+1}- d_i\}$ for increasing sequence $d_1< d_2 <\cdots< d_k$ in the interval $[0,1]$. If $d_i$s are uniformly distributed in $[0,1]$, then it expects $\delta\rightarrow 0$ as $k\rightarrow \infty$. Every finite UGD sampling $S$ of $U_n$ partitions $U_n$ into finitely many disjoint sets $\sigma_1,\cdots, \sigma_m$ such that their measures' sum is equal to the measure of $U_n$. Let $\mu(\sigma_i)$ denote the measure of $\sigma_i$ for $1\leq i\leq m$, set $\sigma:=\max\limits_{1\leq i\leq m}\mu(\sigma_i)$.
Given a Gaussian distribution with density function $p(\mathbf x)$, and a continuous  function $f:\mathbb R^n \mapsto \mathbb R^n$, then
\begin{eqnarray}\label{eq:DInt}
\lim_{\delta\rightarrow 0, \sigma\rightarrow 0} \sum{\omega_{kj}f(\mathcal X_{kj})} = \int{f(\mathbf x)p(\mathbf x) d\mathbf x}
\end{eqnarray}
if the right hand's integral exists, where $\mathcal X_{kj}$ are sampled by GUS. Intuitively, $\delta\rightarrow 0$ means the probability measure set is well partitioned, which are presented by the importance function and $S_=(d_k)$. Such partition could be seen at the radial direction. The composed sets of radial partition could be further partitioned on the spherical direction, which are symmetrically refined by the samples on the ellipsoids $S_=(d_k)$. Then the formula (\ref{eq:DInt}) immediately follows by the means of Lebesgue-Stieltjes integration. This is the theoretical foundation of GUS. It implies the approximation can be as accurate as possible if there are adequate samples taken in such manner.

Another important feature of GUS is that the computational complexity of GUS can be controlled in some acceptable levels. The main computations are related to (\ref{eq:RPC}) for $r_k$s, the Cholesky matrix decomposition $\sqrt{\mathbf{P}_{X}}$ of $\mathbf{P}_{X}$, the weights by (\ref{Eq:NW}) and the stretch scalar $\beta$ by (\ref{eq:SC}). However, each computing step for $r_k$ can be efficiently carried out by some numerical methods.  Moreover, weight and stretch scalar computing are real arithmetic. These mean that the GUS sampling can be quickly done once it has a careful selection of the number $N$. For the sake of efficiency, it often requires $N$ being some polynomial-size of the dimension $n$. Under this requirement, the GUS could be implemented in polynomial time. In the next section, through a target tracking problem, the GUS achieves high accuracy and reliability with practical efficiency.

\section{Simulation Case Study}\label{Simu}
In this section, we report the simulation results by applying the GUF to a target tracking problem derived from \cite{bar2004estimation}, which was used as a benchmark problem in \cite{2009Cubaturekalmanfilters,2013High-degreecubatureKalman} to validate the performance of filters. This problem consider a typical air-traffic control, wherein an aircraft executes maneuvering turn in a horizontal plane at a constant but unknown rate $\Omega$.  The kinematics of the turning motion can be modeled by:
\begin{eqnarray}
\mathbf{x}_{k}&=&\left[
    \begin{array}{ccccc}
     1&\displaystyle\frac{\sin(\Omega \Delta t)}{\Omega}&0&\displaystyle\frac{\cos(\Omega\Delta t)-1}{\Omega}&0\\
     0&\displaystyle\cos(\Omega\Delta t)&0& \displaystyle-\sin(\Omega\Delta t)&0\\
     0&\displaystyle\frac{1-\cos(\Omega\Delta t)}{\Omega}& 1 & \displaystyle\frac{\sin(\Omega\Delta t)}{\Omega}&0\\
     0&\displaystyle\sin(\Omega\Delta t) & 0 & \displaystyle\cos(\Omega\Delta t) & 0\\
     0&0&0&0&1
    \end{array}\right]\nonumber\\
&~&\times \mathbf{x}_{k-1}+\mathbf{v}_{k-1}
\end{eqnarray}
where $\mathbf{x}_{k}=[x_{k},\dot{x}_{k},y_{k},\dot{y_{k}},\Omega]^{T}$ is the state of the aircraft; $x_{k}$ and $y_{k}$ represent the positions, $\dot{x}_{k}$ and  $\dot{y_{k}}$ are the velocities, in two coordinates, at time $k$, respectively; $\Omega$ is the unknown turn rate; $\Delta t$ is time interval between two consecutive measurements; $\mathbf{v}_{k-1}$ is the Gaussian white process noise with its mean zero and covariance $\mathbf{Q}_{k-1}=\mbox{diag} [q_1M \quad q_1 M \quad q_2]$, where
\begin{eqnarray}
M =\left[
    \begin{array}{cc}
    \displaystyle\frac{\Delta t^{3}}{3}& \displaystyle\frac{\Delta t^2}{2}\\
     \displaystyle\frac{\Delta t^2}{2}&\Delta t
     \end{array}\right]\nonumber
\end{eqnarray}
and the scalar parameters $q_1$ and $q_2$ are related to process noise intensities. The measurements are the range $r$ from the origin of the plane, where a radar is equipped, to the location of aircraft, and the bearing, $\theta$. Correspondingly, the measurement equation is
\begin{eqnarray}
   \left( \begin{array}{c}
          r_{k}\\
       \theta_{k}
        \end{array} \right)&=&\left[\begin{array}{c}
       \displaystyle\sqrt{x^{2}_{k}+y^{2}_{k}}\\
        \tan^{-1}(\displaystyle\frac{y_{k}}{x_{k}})
        \end{array}\right]+\mathbf{w}_{k}
\end{eqnarray}
where $\mathbf{w}_{k}$ is the Gaussian white measurement noise with mean zeros and covariance $\mathbf{R}_{k}=\mbox{diag} [\sigma_r \quad \sigma_\theta]$.

To evaluate various nonlinear filter performances, we employ the root mean square error (RMSE) of the position, velocity and turn rate.  For a general and fair comparison, $50$ independent Monte Carlo runs are taken in each filtering process. The RMSE in position at time $k$ is defined by
\begin{eqnarray} \label{Eq:RMSE_Pos}
\mbox{RMSE-Position}(k)=\sqrt{\frac{1}{N}\sum^{N}_{i=1}[(x^{i}_{k}-\bar x^{i}_{k})^{2}+(y^{i}_{k}-\bar y^{i}_{k})^{2}]}
\end{eqnarray}
where $(x^{i}_{k},y^{i}_{k})$ and $(\bar x^{i}_{k},\bar y^{i}_{k})$ are the true and estimated positions at the $i$-th Monte Carlo run at time $k$. Similarly to the RMSE
in position, we may also define the RMSE
\begin{eqnarray} \label{Eq:RMSE-Vel}
\mbox{RMSE-Velocity}(k)=\sqrt{ \frac{1}{N}\sum^{N}_{i=1}[(\dot{x}^{i}_{k}-\bar{ \dot{x}}^{i}_{k})^{2}+(\dot{y}^{i}_{k}-\bar{ \dot{y}}^{i}_{k})^{2}]}
\end{eqnarray}
in velocity and  the RMSE
\begin{eqnarray} \label{Eq:RMSE-TR}
\mbox{RMSE-Turn Rate}(k)=\sqrt{ \frac{1}{N}\sum^{N}_{i=1} (\Omega^{i}_{k}-\bar \Omega^{i}_{k})^{2} }
\end{eqnarray}
in turn rate.

Note that, all filtering algorithms were coded with MATLAB (2010a version) and ran on a computer platform with Intel(R) Core(TM) i3-2100 CPU @ 3.10 GHz and RAM 2.00 GB.

\begin{figure}[!t]
	\centering
	\subfloat[Position]{\includegraphics[width=0.5\textwidth]{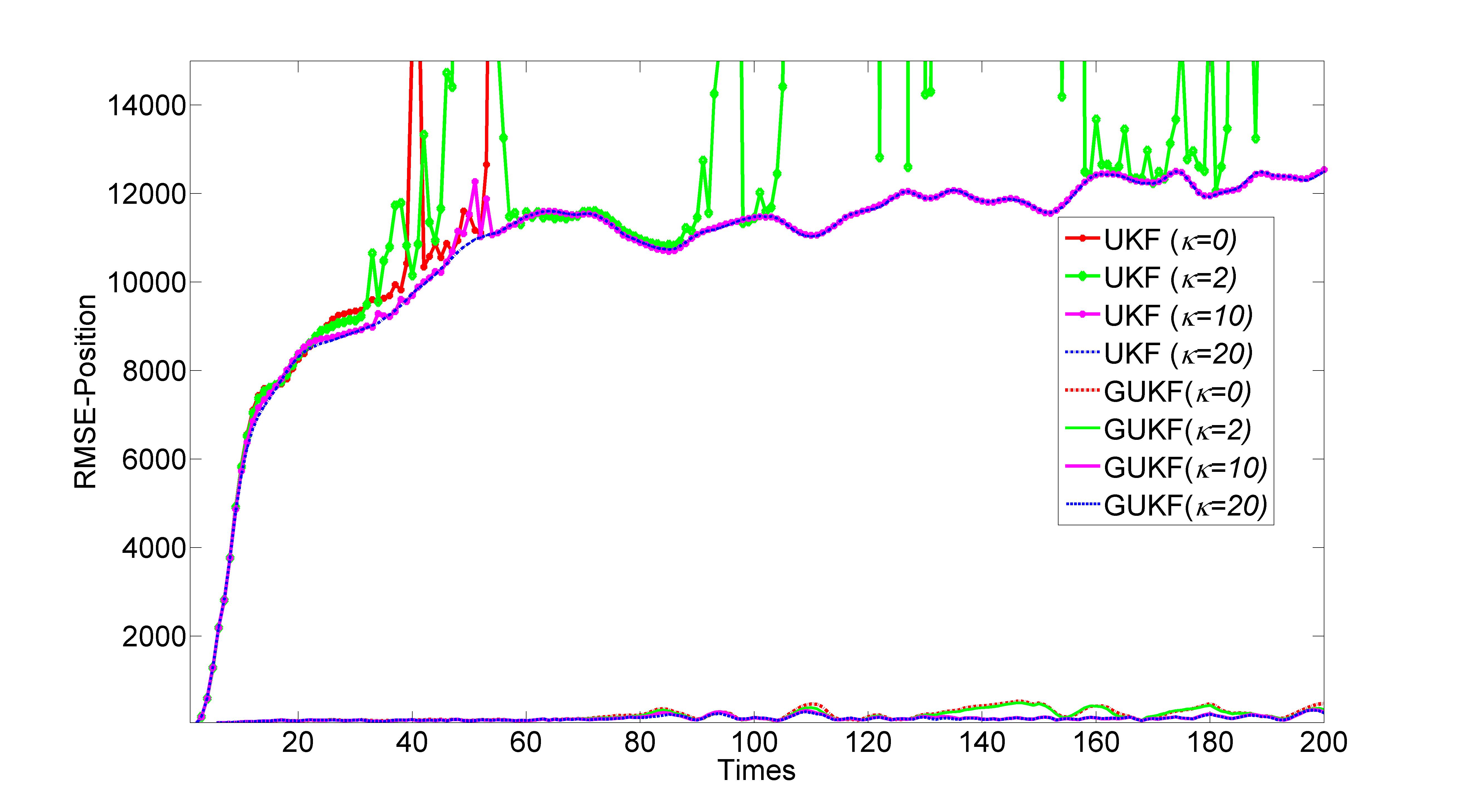}} \hfill
	\subfloat[Velocity]{\includegraphics[width=0.5\textwidth]{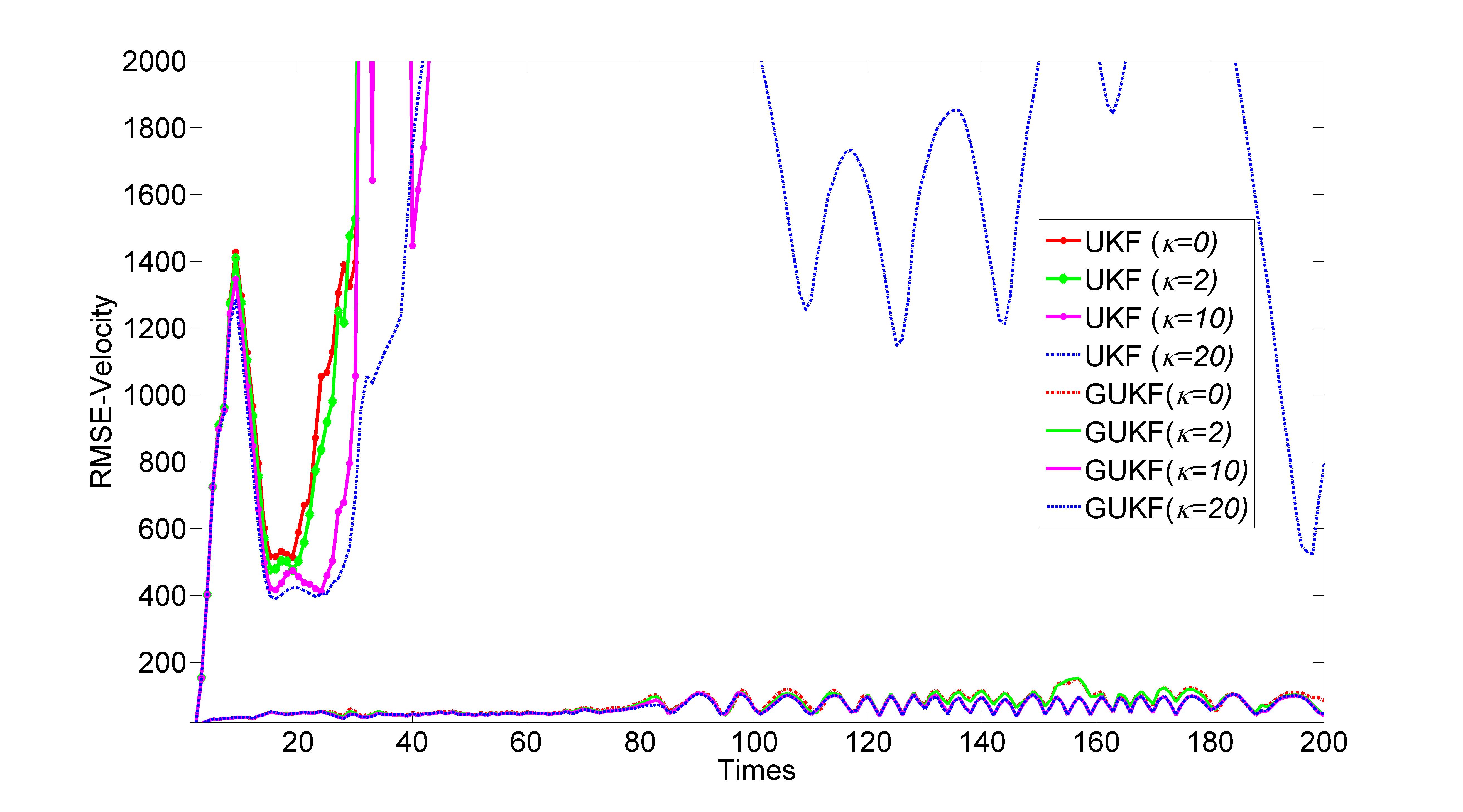}} \hfill
	\subfloat[Turn Rate]{\includegraphics[width=0.5\textwidth]{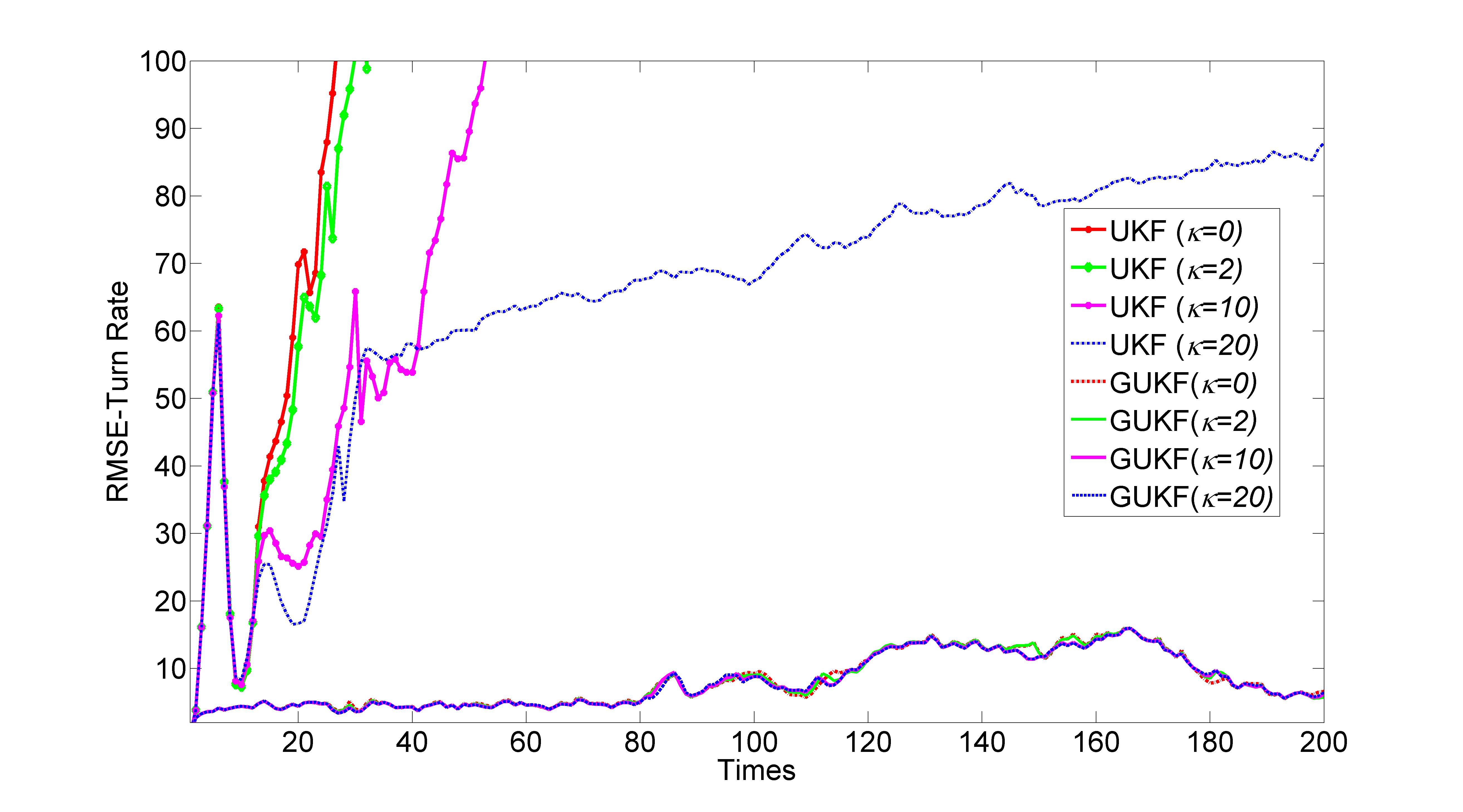}} \hfill
	\caption{The performances of UKF and GUKF}
	\label{fig:UKFsPosition}
\end{figure}

{\it Scenario 1}: For brevity, the UKF adopting the resampling process like CKF under Gaussian assumption is called Gaussian UKF (GUKF).  In section \ref{sec:SSR}, we noted that under the Gaussian assumption, the performance of GUKF should be better than the original UKF. In the following, we use the above target tracking problem to exemplify this view. The data are:
\begin{eqnarray}
&   & \kappa=1,  \nonumber \\
&   & \Omega=-3^{\circ} \mbox{s}^{-1} \nonumber \\
&   & \Delta t=1 \mbox{s}  \nonumber\\
&   & q_1=1 \mbox{m}^2 \mbox{s}^{-3}  \label{sim:data}\\
&   &  q_2=1.75\times 10^{-3} \mbox{s}^{-3} \nonumber \\
&   & \sigma_r= 1000 \mbox{m}^2 \nonumber \\
&   & \sigma_{\theta}= 100 \mbox{m}\mbox{rad}^2 \nonumber
\end{eqnarray}
where $\kappa$ in the parameter in (\ref{UT1}). At the time $k=0$, the estimation of state $\mathbf{x}_{0|0}$ and covariance $\mathbf{P}_{0|0}$ are chosen equally to the initial value
\[ \mathbf{x}_{0}=[1000\mbox{m}, 300\mbox{ms}^{-1}, 1000\mbox{m}, 0\mbox{ms}^{-1},  -3^{\circ}\mbox{s}^{-1}]^{T}\] and
\[ \mathbf{P}_{0}=\mbox{diag}[1000\mbox{m}^{2}, 10\mbox{m}^{2}/\mbox{s}^{2}, 100\mbox{m}^{2},  10\mbox{m}^{2}/\mbox{s}^{2}, 100\mbox{mrad}^{2}/\mbox{s}^{2}]\] respectively. The $\mathbf y_{k+1}$ in (\ref{formula:datafusion}) is generated by a simulating process. All the filters are initialized with the same condition in each run. In each run, the simulation length is 200.

Fig.\ref{fig:UKFsPosition} shows the performances of UKF and GUKF under different parameters $\kappa$. GUKF has better accuracy than UKF.

{\it Scenario 2}: To test the general performance of GUF, we execute it with different choices of the parameter $N$ and the reference sampling. The parameter $N$ roughly determines the distribution of important values $d_k$ and so the weighting values $\sum_{j}\omega_{k,j}$ for $k\leq N$. Together with the number $N_k$ of samples $\mathcal X^*_{k,j}$, each weighting value $\omega_{k,j}$ is fully determined. Under these conditions, the parameter $\beta$ can be computed by (\ref{eq:SC}), and thus the samples are figured out.

In this scenario, all systematic parameters are as same as in Scenario 1,  including (\ref{sim:data}) and the initial values. Once the  $N$ is fixed, we make use of uniformly distributed numbers \begin{eqnarray}
\label{Eq:ISALPHASIM}
d_k =\frac{k}{N+1}, ~~~k=1,2,\cdots,N
\end{eqnarray}
Then we apply $d_k$ to (\ref{eq:RPC}) to get $r_k$ and the corresponding normalized weights.
Here, we adopt (\ref{NumR}) to choose same amounts $N_k$ of samples for each $r_k$.  For brevity, let $\Theta^n$ stand for the sign change operators and $\Phi^n$ denote the set of all permutation operators on the coordinates of an $n$-dimension vector. We carry out four GUF, denoted by GUFi for $1\leq i\leq 4$, over the following parameters:

In GUF1, $N=1$, the reference sampling is
\begin{eqnarray} \label{Eq:RS1}
S_k=\{\phi \circ \theta (\mathcal S) | \mathcal S=(1,0,0,0,0), \phi\in \Phi^5, \theta\in \Theta^5\}
\end{eqnarray}
and hence  $N_k=|S_k|=10$ for $k=1,2,3$. There are totally $10$ samples in each sampling.

In GUF2, $N=2$, the reference sampling is
\begin{eqnarray} \label{Eq:RS2}
S_k=\{\phi \circ \theta (\mathcal S) | \mathcal S=(1,1,0,0,0), \phi\in \Phi^5, \theta\in \Theta^5\}
\end{eqnarray}
and hence  $N_k=|S_k|=50$  for $k=1,2,3$. There are totally $100$ samples in each sampling.

In GUF3, $N=7$, the reference sampling is
\begin{eqnarray} \label{Eq:RS3}
S_k=\{\phi \circ \theta (\mathcal S) | \mathcal S=(1,1,1,0,0), \phi\in \Phi^5, \theta\in \Theta^5\}
\end{eqnarray}
and hence  $N_k=|S_k|=130$  for $k=1,2,3$. There are totally $910$ samples in each sampling.

In GUF4, $N=9$, the reference sampling is
\begin{eqnarray} \label{Eq:RS4}
S_k=\{\phi \circ \theta (\mathcal S) | \mathcal S=(1,1,1,1,0), \phi\in \Phi^5, \theta\in \Theta^5\}
\end{eqnarray}
and hence $N_k=|S_k|=210$  for $k=1,2,3$. There are totally $1890$ samples in each sampling.

\begin{figure}[!t]
\centering
\subfloat[Position]{\includegraphics[width=0.5\textwidth]{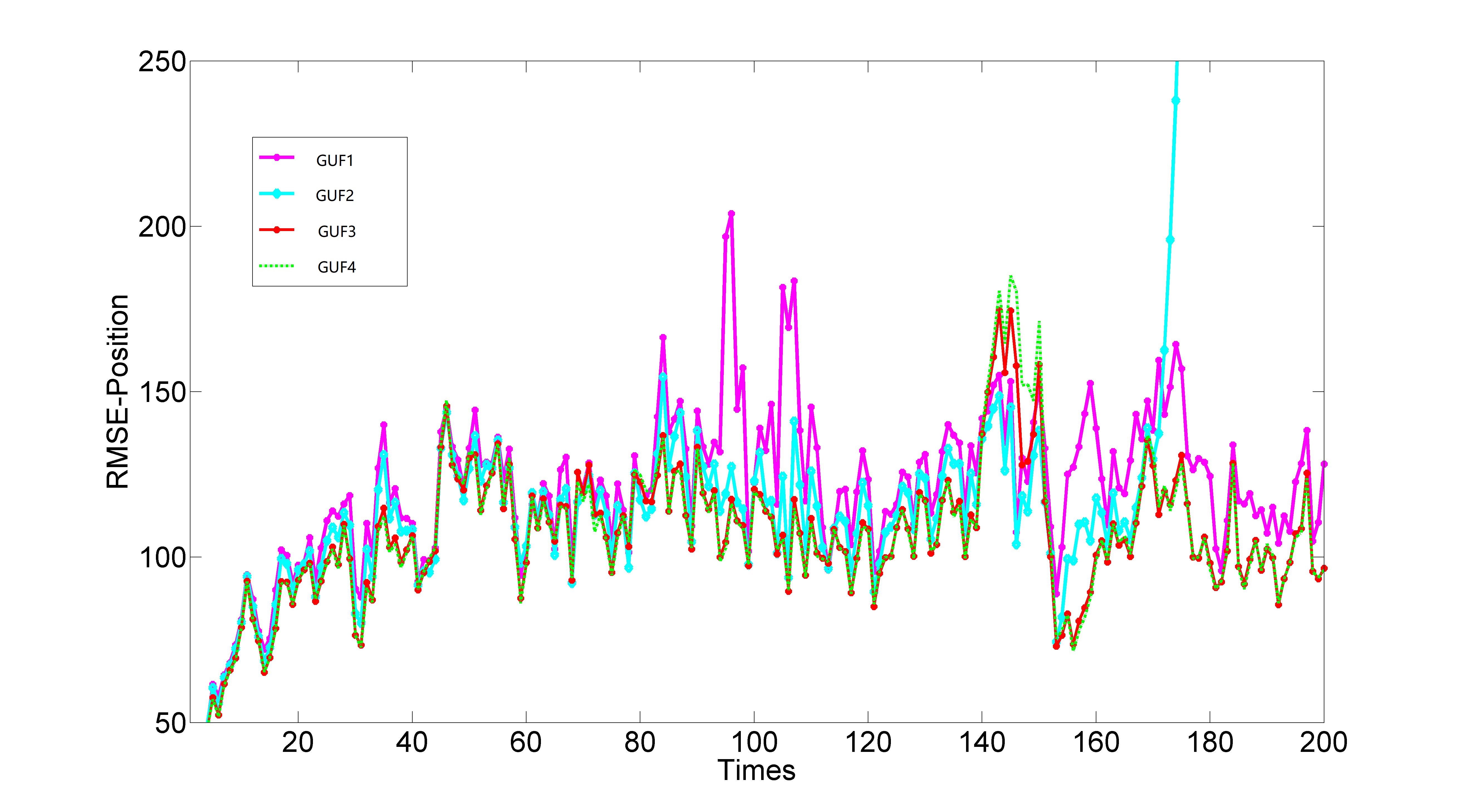}} \hfill
\subfloat[Velocity]{\includegraphics[width=0.5\textwidth]{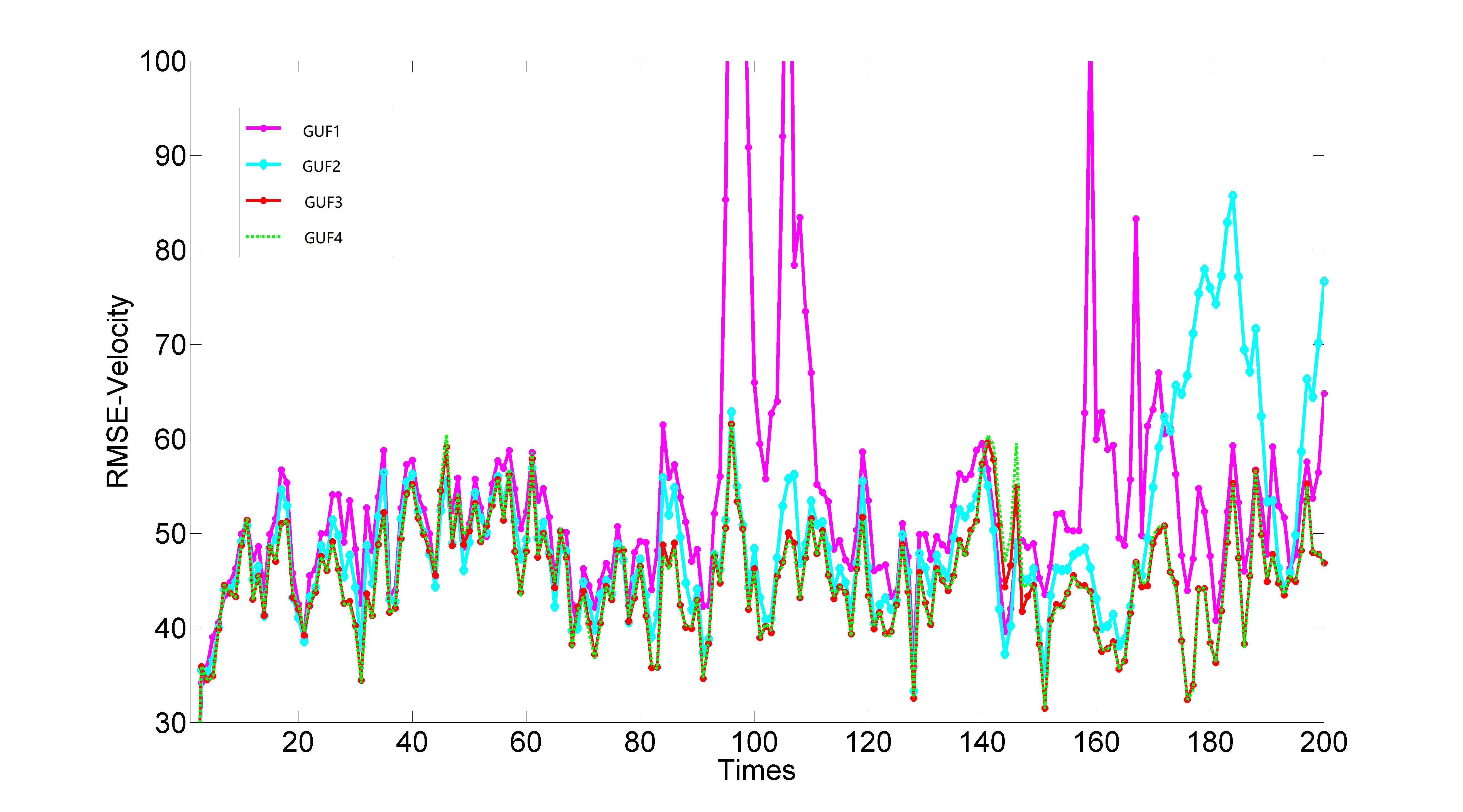}} \hfill
\subfloat[Turn Rate]{\includegraphics[width=0.5\textwidth]{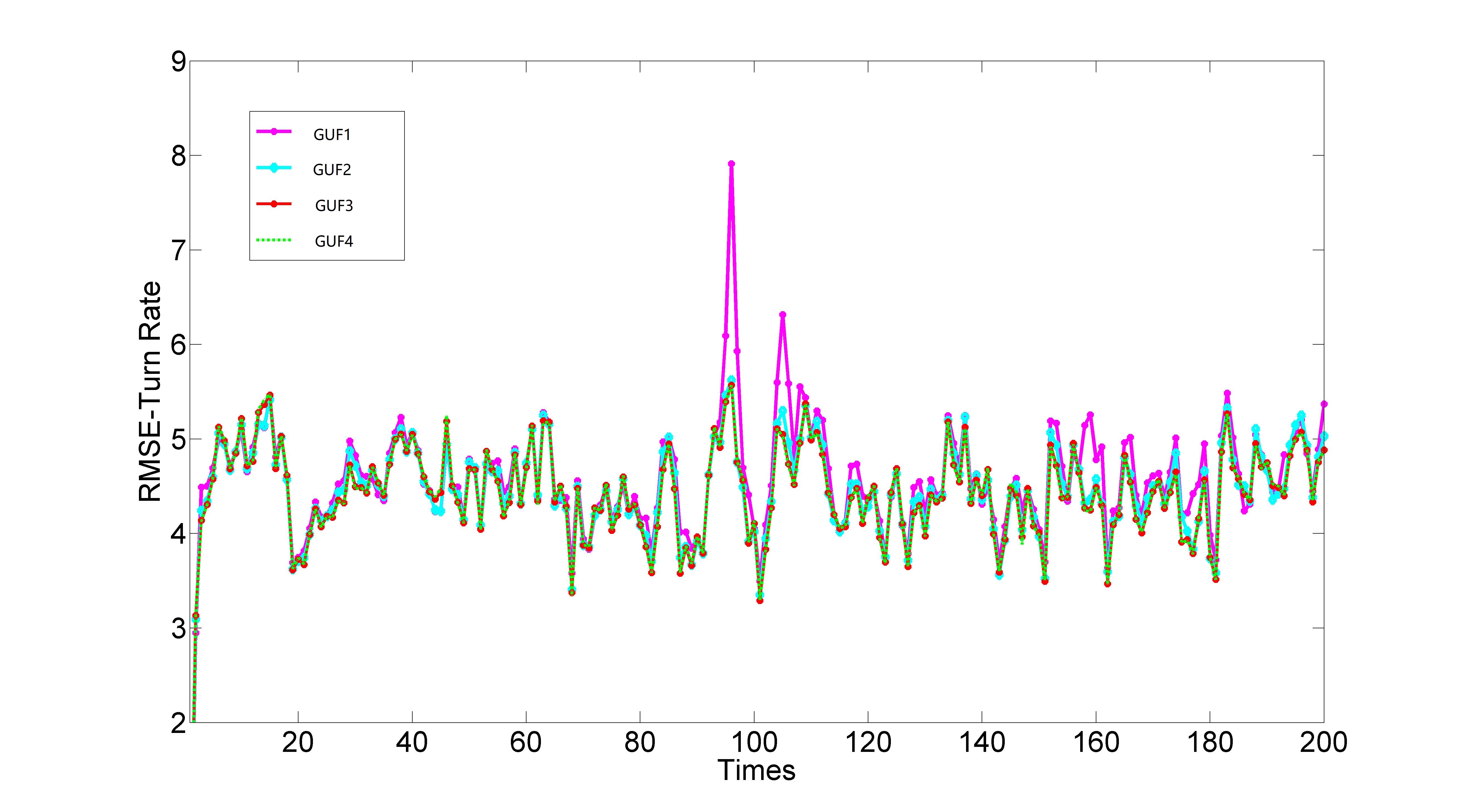}} \hfill
  \caption{The GUFs with different choices of sample}
  \label{fig:GUFs}
\end{figure}

Fig.\ref{fig:GUFs} demonstrates the performances of these four GUF implementations.  Roughly, the accuracy of GUF increases with the number of samples. This is a nature of unscented method. As seen in GUF,  the $d_k$s are sampled by this method. Anyway, the average error is quite stable in terms of RSME. This is due to the special distribution of samples derived from reference samples.

\begin{figure}[!t]
\centering
\subfloat[Position]{\includegraphics[width=0.5\textwidth]{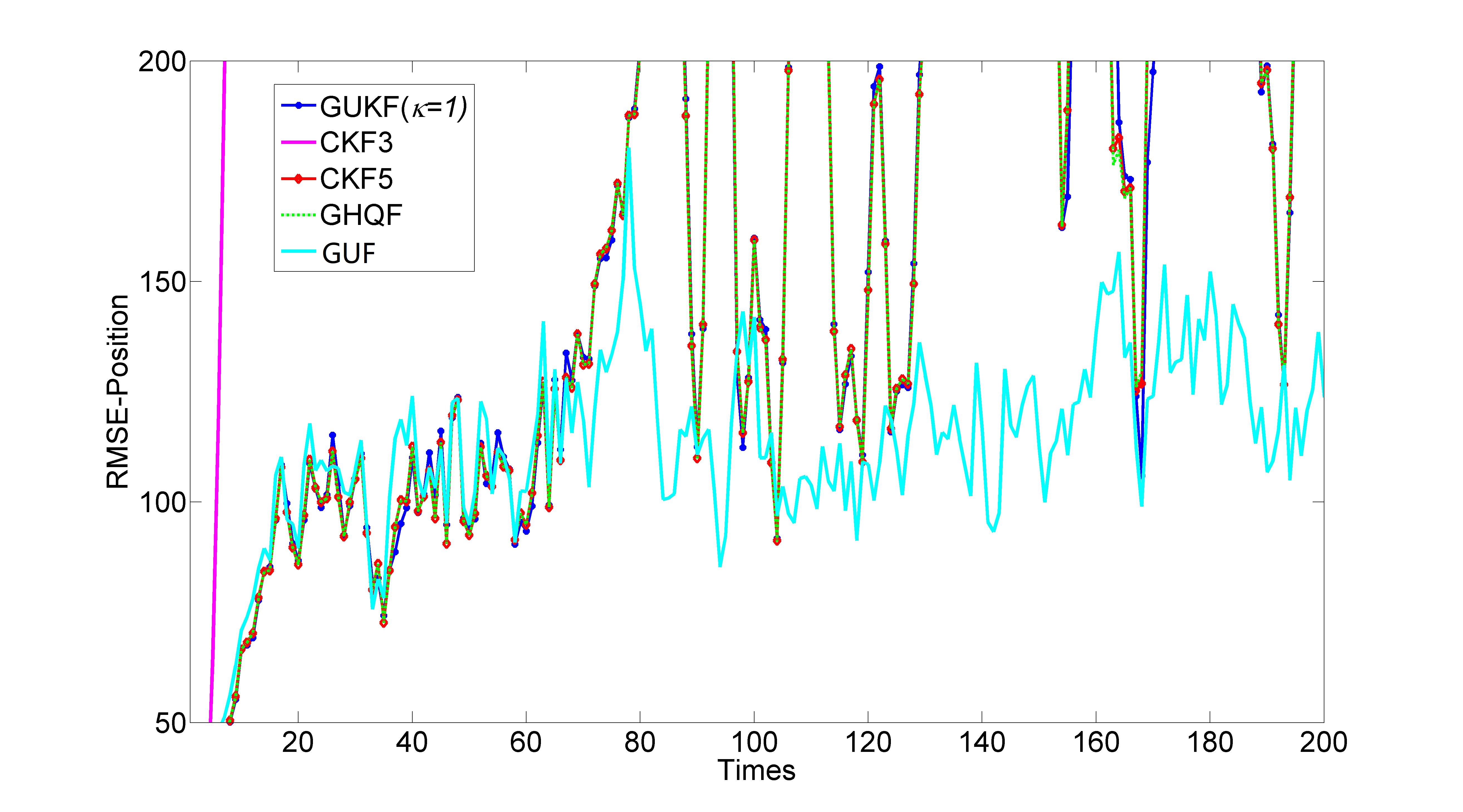}} \hfill
\subfloat[Velocity]{\includegraphics[width=0.5\textwidth]{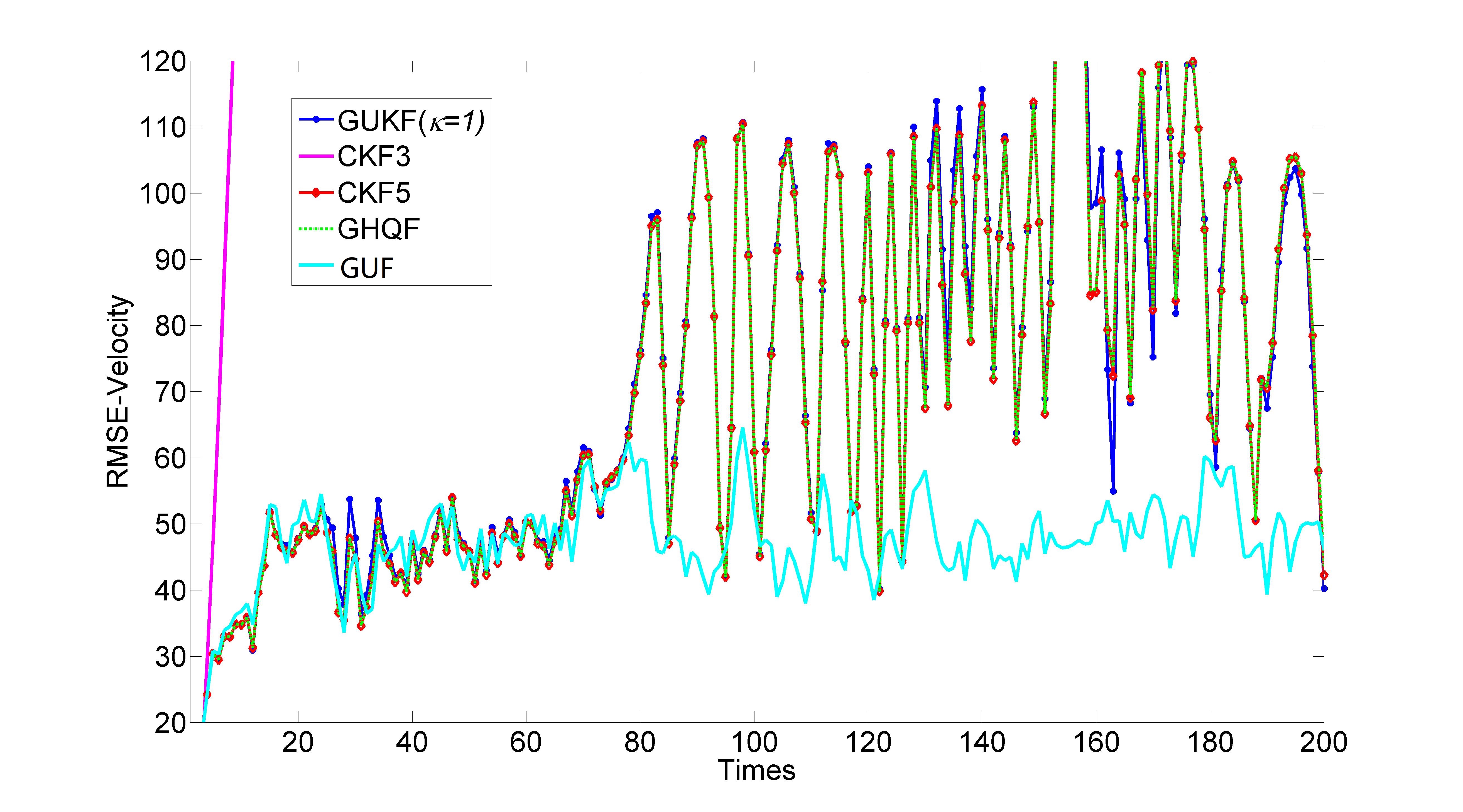}} \hfill
\subfloat[Turn Rate]{\includegraphics[width=0.5\textwidth]{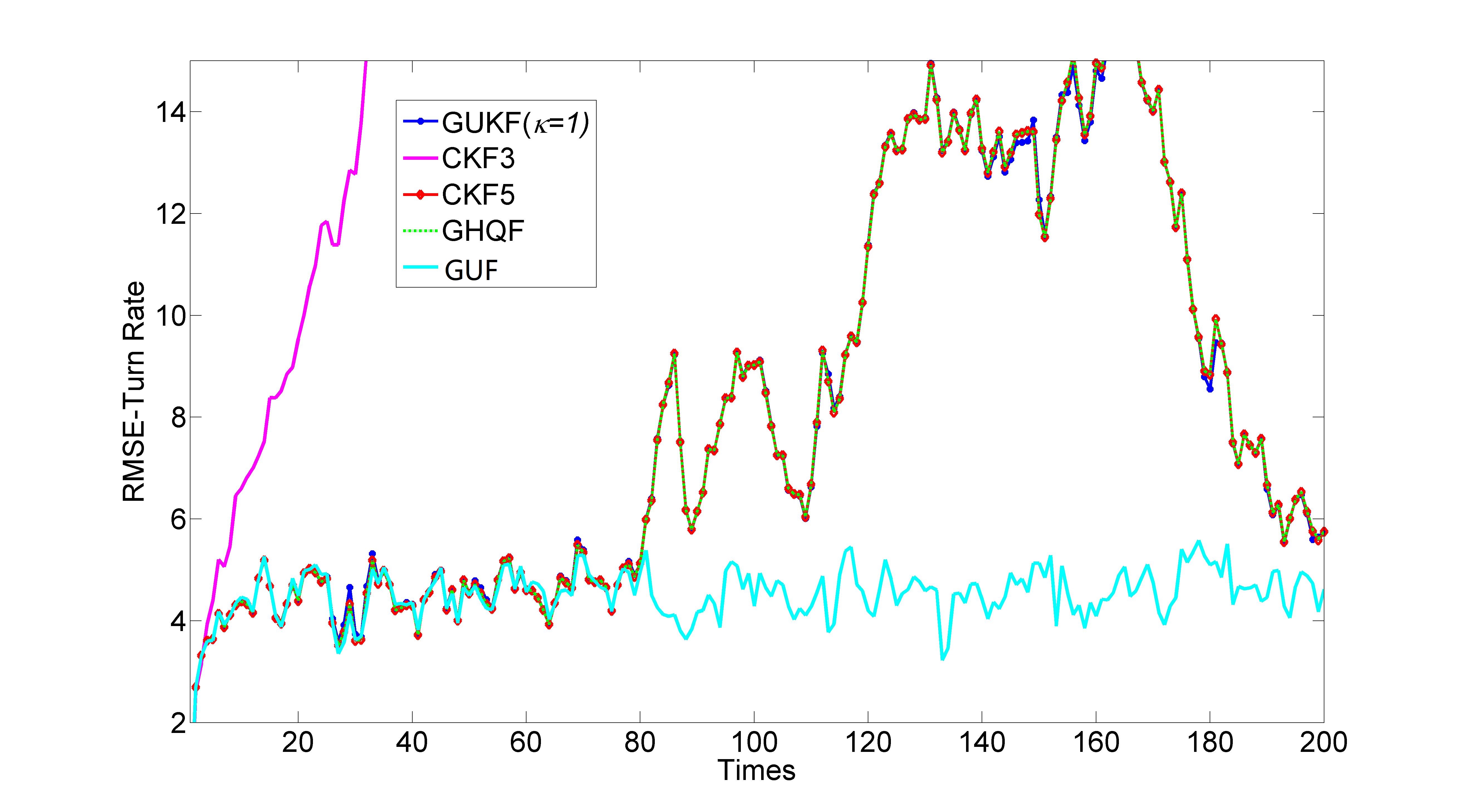}} \hfill
  \caption{The comparison of nonlinear Kalman filters: GUKF, CKF3, CKF5, GHQF, GUF}
  \label{fig:DNKF}
\end{figure}

\begin{figure}[!t]
\centering
\subfloat[Position]{\includegraphics[width=0.5\textwidth]{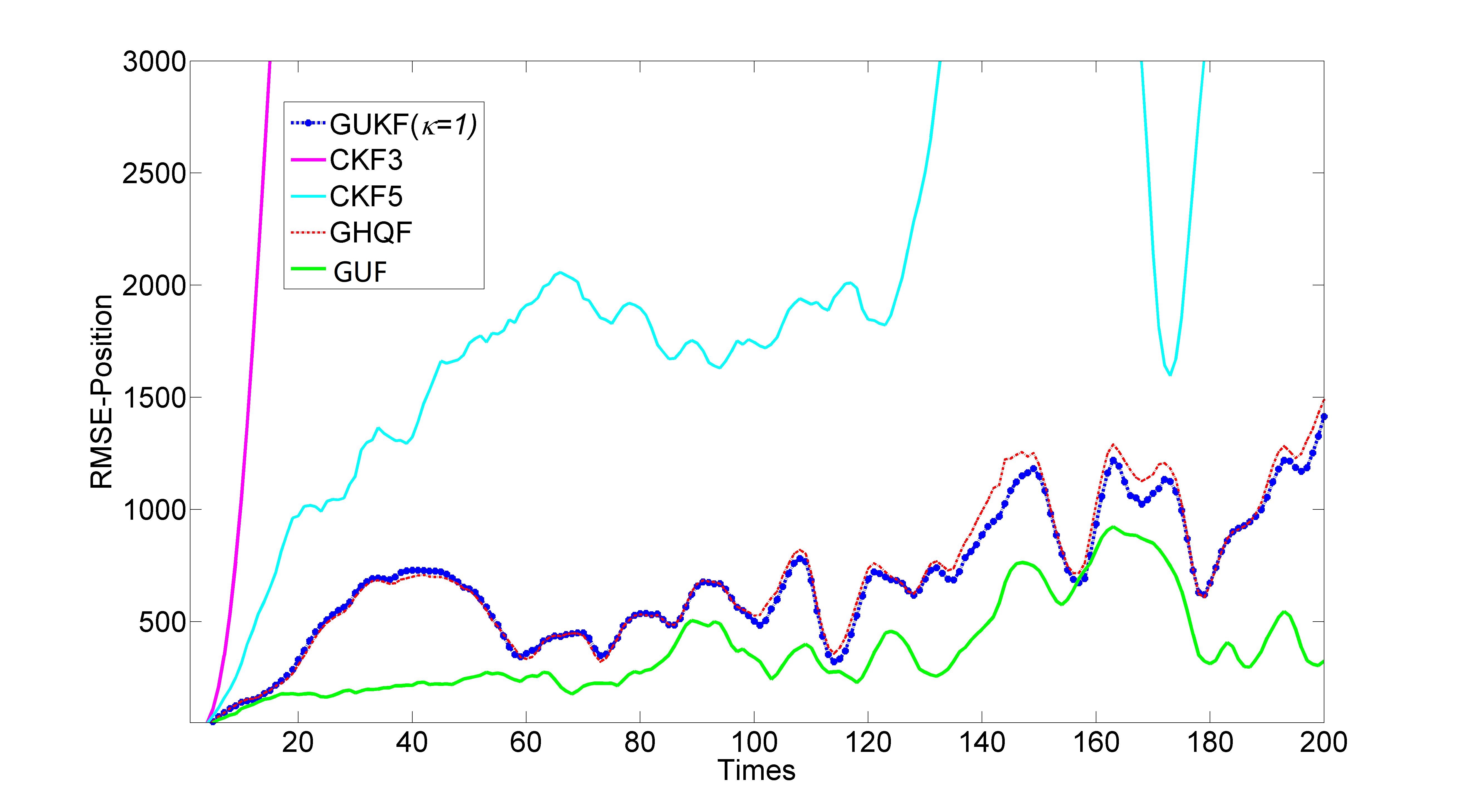}} \hfill
\subfloat[Velocity]{\includegraphics[width=0.5\textwidth]{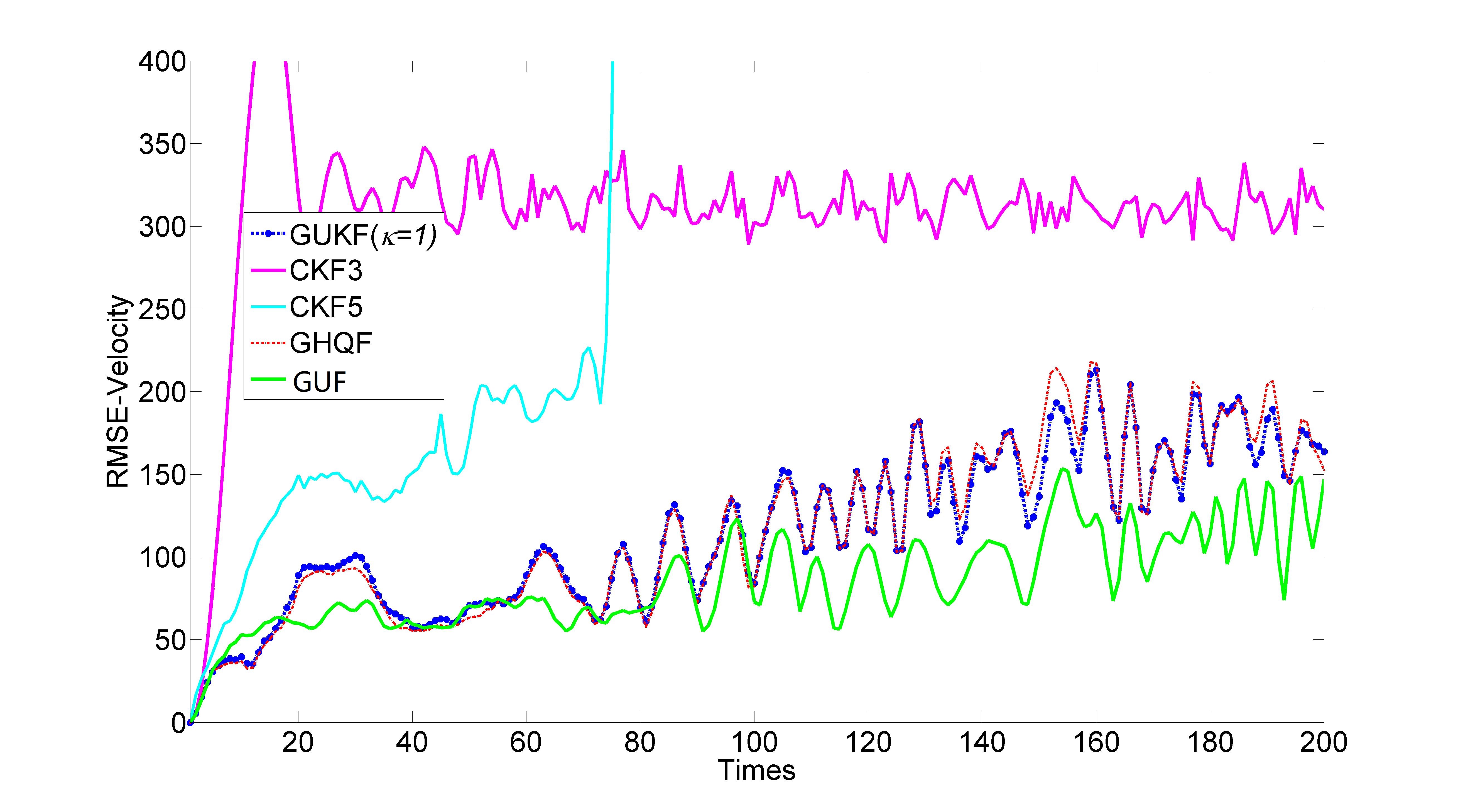}} \hfill
\subfloat[Turn Rate]{\includegraphics[width=0.5\textwidth]{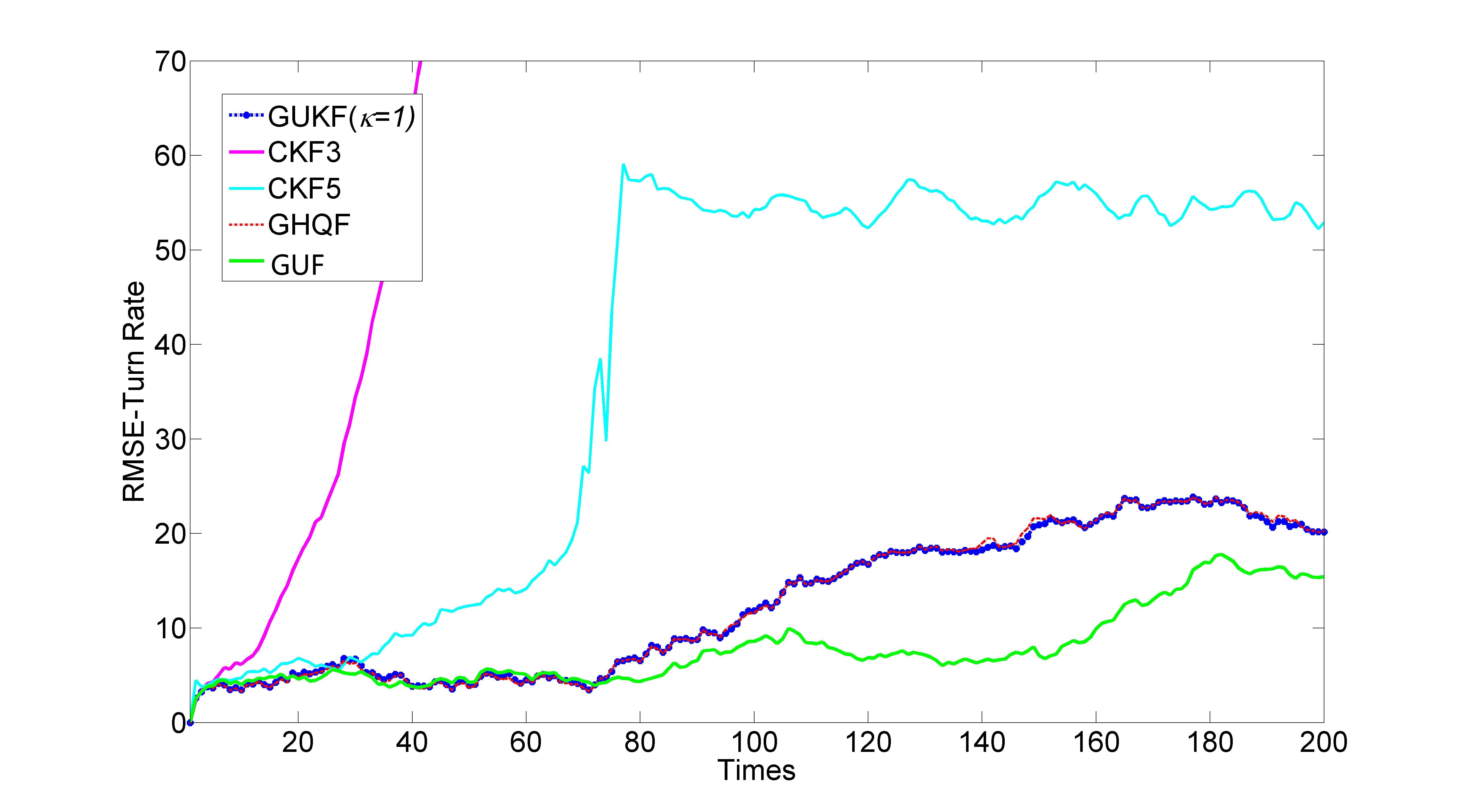}} \hfill
  \caption{The performances of nonlinear Kalman filter GUKF, CKF3, CKF5, GHQF and GUF, when the Gaussian nature of the problem is explicitly violated}
  \label{fig:NGNF}
\end{figure}

{\it Scenario 3}: To demonstrate the supremacy of GUF at some aspect over the existing filters, we ran following nonlinear filters: the GUKF, the three degrees CKF (CKF3) \cite{2009Cubaturekalmanfilters}, the three degrees CKF (CKF5) \cite{2013High-degreecubatureKalman},  the quadrature GHF (GHQF) \cite{2000Gaussianfiltersfornonlinear} and the GUF, with the same system parameters as in Scenario 1,  including (\ref{sim:data}) and the initial values. Note that, the systematic data here used is the same as the data for PF in \cite{2013High-degreecubatureKalman}, under which PF need much time and a large number of samples to achieve similar accuracy to CKF3 or CKF5. As it was already studied in such a way, we skip repeating PF in the comparison.

Here, we select GUKF instead of UKF, since it has better accuracy, as seen in Scenario 2. The GUKF using the referred value $\kappa=3-n$ by \cite{2004Unscentedfiltering} was found often halt its operation due to the indefinity of covariances in highly nonlinear and uncertain systems when $1\leq \kappa\leq 10$, the GUKF here can achieve a good performance.  As these performances are quite close to each other, without loss of generality, we pick $\kappa=1$ in the comparison study.  For a similar reason, in GUF, we pick $N=2$ and the reference samples specified by (\ref{Eq:RS1}). For GUKF, CKF3, CKF5, GHQF and GUF, the numbers of the sample are summarized in the second row of Table \ref{tab:NoS}. In this table, it reports the computation durations in the third row for each filter, including the RSME graph plots showed by Fig.\ref{fig:DNKF}. From the time consuming, we can see that GUF can maintain acceptable efficiency.

\begin{table}[!h]
\renewcommand{\arraystretch}{1.3}
\caption{The Numbers of Sample}
\label{tab:NoS}
\centering
\begin{tabular}{rccccc}
\hline
 \vline &  GUKF & CKF3 & CKF5 & GHQF & GUF \\
\hline
Samp.No.  \vline & 11 & 10 & 51 &  243 & 20 \\
\hline
Run-Time  \vline & 2.5395 & 2.5330 & 7.3599 &  35.0053 & 4.5268 \\
\end{tabular}
\end{table}

Fig. \ref{fig:DNKF} depicts the performances of these filters. As can be seen in \cite{2013High-degreecubatureKalman}, our simulation also showed that CKF5 and GHQF maintain an indiscernibly comparable performance to each other. Besides, GUKF with $\kappa=1$ also maintains an indiscernibly relative performance to them. However, their RSMEs fluctuate greatly in a wider range than GUF. This confirms that GUF can have better computational stability than existing methods at reasonable computational complexity.

{\it Scenario 4}: To test the robustness of the GUF, we extend the previous target tracking problem to  the case when the Gaussian nature of the problem is explicitly violated as follows. Let the measurement noise  $\mathbf w_k$ follow a Gaussian mixture
\begin{eqnarray}\label{Eq:MixGD}
\mathbf w_k & \sim & 0.5 \mathcal N(\mathbf 0, R_1)+0.5 \mathcal N(\mathbf 0, R_2) \quad  \mbox{with}\\
R_1 & = & \left(\begin{array}{cc}
1000\mbox{ m}^2   & 150\mbox{ m}^2\mbox{rad}  \\
150\mbox{ m}^2\mbox{rad} & 100 \mbox{ m rad}^2  \\
\end{array}\right)  \\
R_1 & =& \left(\begin{array}{cc}
50 \mbox{ m}^2   &  100 \mbox{ m}^2\mbox{rad}  \\
100 \mbox{ m}^2\mbox{rad} & 1000 \mbox{ m rad}^2  \\
\end{array}\right)
\end{eqnarray}
This setup is taken from \cite{2013High-degreecubatureKalman}, which is similar to a setup in \cite{2009Cubaturekalmanfilters}. Other systematic parameters are set as the same as before. In this scenario, we took $N=5$ and  the reference samples specified by (\ref{Eq:RS1}) for GUF. The RSMEs for different filters are shown in Fig.\ref{fig:NGNF}. Again, the GUF outperforms other filters roughly in the accuracy and computational stability. It has almost same number of samples as the corresponding number in CKF5, which is much less than the corresponding number in GHQF. So its computational complexity is also acceptable.

The simulation results exemplify our initial motivation of the GUF that avoid negative weights and improve accuracy by increasing samples with reasonable computational costs. Moreover, this indicates that for the sake of accuracy, efficiency and stability, the GUF is a good candidate for nonlinear Kalman filters, especially for the systems of higher dimensions, acute nonlinearity and high degrees of uncertainty.

\section{Conclusion}\label{Con}

In this article, we have proposed a new nonlinear Kalman filter called geometric unscented filter and illustrated this filtering under the Gaussian assumption. Note that, the GUF is a general framework for nonlinear systems. The  Gaussian assumption is used in the article only for the sake of easy understanding. Anyway, the GUF is inspired by PF, UKF and CKF in terms of sampling and filtering. As for sampling, it makes use of the idea of importance sampling in PF, the moments matching in UKF and the massive, symmetric sampling in CKF. Using moments matching captures the main characters, e.g. mean and covariance, of a probability distribution. With the massive, symmetric sampling and the IF derived from PF together, instead of the higher-order moments matching in UKF and CKF, it improves the accuracy at a reasonable computational cost. As to the filtering, it adopts the famous Kalman filtering, which is also the filtering framework of UKF and CKF, to obtain the optimal estimation to the least square errors. It could also be seen as a simplified PF with the special resampling strategy, namely the GUS to avoid the dimension curse in PF. Summarily, the GUF is a scalable and semi-deterministic sampling method as a selective mixture of PF, CKF and UKF, drawing advantage of them such as the positive probability weights of PF, the crucial probability information (mean and covariance) catching in UKF and the efficient sampling of CKF.

\appendix[Computing the coefficient $H_n(\mathcal B)$ in  \ref{sec:NWandMM}]\label{appendix:proofs}
In this appendix, we provide the computation of the coefficient $H_n(\mathcal B)$ in  \ref{sec:NWandMM}. To this end, we need to introduce some notations and notions as follows. Let $\Theta$ stand for the sign change operators and $\Phi$ denote the set of all permutation operators on the coordinates of an $n$-dimension vector. By the choice of UGD samples on $U_n$, if $\mathcal S\in U_n$ then $\phi\circ\theta (\mathcal S)\in U_n$ for any $\phi\in \Phi$ and $\theta\in \Theta$.

Given a UGD  $S$ of $U_n$, a {\it basis} of $S$ is a subset $ B \subseteq S$ satisfying (1) for any two vectors $\mathcal S_1\neq \mathcal S_2\in  B$, called {\it bases} of $S$, $\mathcal S_1\neq \phi\circ\theta(\mathcal S_2)$ for any $\phi\in \Phi$ and $\theta\in \Theta$, and (2) for any $\mathcal S\in S$, there are $\mathcal S'\in  B$,  $\phi\in \Phi$ and $\theta\in \Theta$ such that $\mathcal S'= \phi\circ\theta(\mathcal S)$.

The UGDs of  $U_n$ have a nice outer (tensor) product form in what follows. Given a base $\mathbf b$ from a basis $B$ of a UGD $S$ on  $U_n$, in the set $G(\mathbf b):=\{\phi\circ\theta(\mathbf b)\mid \phi\in \Phi, \theta\in \Theta\}$ there must be a member $\mathbf x ( x_1,x_2,\cdots,x_n)^T\in G(\mathbf x)$ such that
\begin{eqnarray}
\label{XC}
x_1=x_2=\cdots=x_{a_1} &\ge& x_{a_1+1}=x_{a_1+2}=\cdots=x_{a_1+a_2}\nonumber\\
&\ge&\cdots\\
&\ge&x_{\mu+1}=x_{\mu+2}=\cdots=x_{\mu+a_M}\nonumber\\
&\ge&0\nonumber
\end{eqnarray}
where $\mu=\sum_{i=1}^{M-1}a_i$, $\sum_{i=1}^{M}a_i=n$, and  $1\le a_i\le n$ for all $i$. Such $\mathbf x$ is called {\it generator}. Moreover, it is evident $G(\mathbf b)=G(\mathbf x)$.  This allows us to choose a basis consisting of generators. Such basis is called {\it standard basis}.

It is clear that the standard basis is chosen in the first quadrant of $n$-dimension Cartesian coordinate system.Then the quantity $N(\mathbf{x})$ of nonzero entries of $\mathbf{x}$ is:
\begin{eqnarray}
\label{NZC}
N(\mathbf{x})=\sum_{i=1}^{n}\mbox{sgn}(x_i)
\end{eqnarray}
where sgn$(\cdot)$  is the sign function.

Let $\Upsilon:=\{\phi\circ\theta\mid \phi\in \Phi ~ \& ~ \theta\in \Theta \}$ and define  $\Upsilon(\mathbf x):=\{\tau(\mathbf x)\mid \tau\in \Upsilon\}$. If all the inequalities in formula (\ref{XC}) are strict, then $\Upsilon(\mathbf x)$ has  $\displaystyle 2^{N(\mathbf{x})}\frac{n!}{\prod_{j=1}^{M}a_i!}$ many members.

For brevity, let $\Theta^n$ stand for the sign changing operators on the coordinates of an $n$-dimension vector. For a $\theta\in\Theta^n$,  $\theta$ is defined by $\theta=(\theta_1,\theta_2,\cdots,\theta_n)$ with   $\theta_i=\pm 1$ such that $\theta(\mathbf{x})=(\theta_1x_1,\theta_2x_2,\cdots,\theta_n x_n)^T$.  And let $\Theta^n(\mathbf x):=\{\theta(\mathbf x)\mid \theta\in \Theta^n\}$ be the images of $\mathbf x$ under the operators of $\Theta^n$. Similarly,  let $\Phi^n$ denote the set of all permutation operators on $n$ objects and $\Phi^n(\mathbf x):=\{\phi(\mathbf x) \mid \phi \in \Phi^n\}$, correspondingly, define $\Upsilon^n:=\{\phi\circ\theta\mid \phi\in \Phi ^n~ \& ~ \theta\in \Theta^n \}$ and $\Upsilon^n(\mathbf x)$.  In the following, we  present the sums of outer products for the sets $\Theta^n(\mathbf x), \Phi^n(\mathbf x)$, and $\Upsilon^n(\mathbf x)$ with different sorts of $\mathbf x$.

\begin{lem}
	\label{Lma1}
	Take $\mathbf{x}=(x_1,x_2,\cdots,x_n)^T \in\mathbb{R}^{n}$.  If for all $1\leq i\leq n,  x_i \neq 0$, then
	\begin{eqnarray}
	\label{CL1}
	\sum_{\mathbf y \in \Theta^n(\mathbf x)}\mathbf y\mathbf y^T=2^n\left[
	\begin{array}{cccc}
	x_1^2&&&\\
	&x_2^2&&\\
	&&\ddots&\\
	&&&x_n^2\\
	\end{array}\right]
	\end{eqnarray}
\end{lem}

\begin{proof} We show this result by induction.  For $n=1$, $\mathbf{x}=(x_1)^T$, then the left side of the equation (\ref{CL1}) is
	\begin{eqnarray}
	\label{PR1}
	\sum_{\mathbf y \in \Theta^1(\mathbf x)}\mathbf y\mathbf y^T=x_1*x_1+(-x_1)*(-x_1)=  2^1x_1^2\nonumber
	\end{eqnarray}
	and so the (\ref{CL1}) is true.
	
	Assume that for $n=k$, for any $\mathbf{x}=(x_1,x_2,\cdots,x_k)^T$ without zero entries, the equation (\ref{CL1}) is true, that is
	\begin{eqnarray}
	\label{PRK}
	\sum_{\mathbf y \in \Theta^k(\mathbf x)}\mathbf y\mathbf y^T=2^k\left[
	\begin{array}{cccc}
	x_1^2&&&\\
	&x_2^2&&\\
	&&\ddots&\\
	&&&x_k^2\\
	\end{array}\right]
	\end{eqnarray}
	
	For $n=k+1$,  take a $\mathbf{x}=[x_1,x_2,\cdots,x_k,x_{k+1}]^T$  without zero entries. For a succinct depiction,  let $\mathbf{x}=[\sigma^T,x_{k+1}]^T$ and $\sigma=[x_1,x_2,\cdots,x_k]^T$. Then $ \displaystyle\sum_{\mathbf y \in \Theta^{k+1}(\mathbf x)}\mathbf y\mathbf y^T$ could be computed as following
	\begin{eqnarray}
	\label{PRn1}
	\sum_{\mathbf y \in \Theta^{k+1}(\mathbf x)}\mathbf y\mathbf y^T
	&=&\sum_{\mathbf y \in \Theta^{k}(\sigma)} \left\{[\mathbf y^T,x_{k+1}]^T[\mathbf y^T,x_{k+1}]\right.\nonumber\\
	&&+\left.[\mathbf y^T,-x_{k+1}]^T[\mathbf y^T,-x_{k+1}]\right\}\nonumber\\
	&=&\sum_{\mathbf y \in \Theta^{k}(\sigma)}\left\{
	\left[\begin{array}{cc}
	\mathbf y\mathbf y^T&\mathbf yx_{k+1}\\
	x_{k+1}\mathbf y^T&x_{k+1}^2\\
	\end{array}\right]\right.\nonumber\\
	&&+\left.\left[
	\begin{array}{cc}
	\mathbf y\mathbf y^T&-\mathbf y x_{k+1}\\
	-x_{k+1}\mathbf y^T&x_{k+1}^2\\
	\end{array}\right]\right\}\nonumber\\
	&=&\sum_{\mathbf y \in \Theta^{k}(\sigma)}2\left[\begin{array}{cc}
	\mathbf y\mathbf y^T&\\
	&x_{k+1}^2\\
	\end{array}\right]\nonumber\\
	&=&2\left[\begin{array}{cc}
	\sum\limits_{\mathbf y \in \Theta^{k}(\sigma)}\mathbf y\mathbf y^T&\\
	&2^k x_{k+1}^2\\
	\end{array}\right]
	\end{eqnarray}
	Using the assumption (\ref{PRK}), the equality (\ref{PRn1}) is transformed into
	\begin{eqnarray}
	\label{PRN2}
	\sum_{\mathbf y \in \Theta^{k+1}(\mathbf x)}\mathbf y\mathbf y^T=2^{k+1}\left[
	\begin{array}{cccc}
	x_1^2&&&\\
	&x_2^2&&\\
	&&\ddots&\\
	&&&x_{k+1}^2\\
	\end{array}\right]\nonumber
	\end{eqnarray}
	This completes the proof of the Lemma \ref{Lma1}.
\end{proof}

If the vector $\mathbf x$ contains some zero entries, then the sign changes on zero entries make no sense. In such a case, the size of set $\Theta^n(\mathbf x)$ is reduced. Correspondingly,  the coefficient of the diagonal matrix of formula (\ref{CL1})  is reduced. Generally, we have the following result.

\begin{lem}
	\label{Lma2}
	Let $\mathbf{x}=(x_1,x_2,\cdots,x_n)^T\in\mathbb{R}^{n}$, then
	\begin{eqnarray}
	\sum_{\mathbf y \in \Theta^n(\mathbf x)}\mathbf y\mathbf y^T=2^{N(\mathbf{x})}\left[
	\begin{array}{cccc}
	x_1^2&&&\\
	&x_2^2&&\\
	&&\ddots&\\
	&&&x_n^2\\
	\end{array}\right]
	\end{eqnarray}
\end{lem}

Additionally,  we put the permutation operators into consideration of the sum of  outer products. That is, consider the sum $\sum_{\mathbf y \in \Upsilon^n(\mathbf x)}\mathbf y\mathbf y^T$ for $\mathbf{x}=(x_1,x_2,\cdots,x_n)^T\in\mathbb{R}^{n}$.  Firstly,  $\Upsilon^n(\mathbf x)$ should have some SB. Furthermore, such SB can be consisted of one base. Without loss of generality, assume that $\mathbf x$ is a generator of  $\Upsilon^n(\mathbf x)$.
\begin{thm}
	If $\mathbf{x}=(x_1,x_2,\cdots,x_n)^T\in\mathbb{R}^{n}$ is a  generator, then it has
	\begin{eqnarray}
	\label{TH2}
	\sum_{\mathbf y \in\Upsilon^n(\mathbf x)}\mathbf y\mathbf y^T  =H_n(\mathbf{x})\mathbf E_n \label{eq:outerpro} \\
	H_n(\mathbf{x})=2^{N(\mathbf{x})}\frac{(n-1)!}{\prod_{j=1}^{M}a_j!}\sum_{i=1}^{M}x_{t_i}^2a_i  \label{eq:couterpro}
	\end{eqnarray}
	where $t_i=\sum_{j=1}^{i}a_j$, $\mathbf E_n$ is the $n$-dimensional identity matrix.
\end{thm}

\begin{proof}

	For a permutation operator $\phi \in \Phi_n$ on vector $\mathbf{x}$,  let $\phi(\mathbf{x})=(\phi(\mathbf{x})_1,\phi(\mathbf{x})_2,\cdots,\phi(\mathbf{x})_n)^T$ be the image, where $\phi(\mathbf{x})_k$ stands for the $k$-th element of vector $\phi(\mathbf{x})$. Then
	\begin{eqnarray}
	\label{General}
	\sum_{\mathbf y \in \Upsilon^n(\mathbf x)}\mathbf y\mathbf y^T =\sum_{\mathbf z \in \Phi^n(\mathbf x)} \sum_{\mathbf y \in \Theta^n(\mathbf z)} \mathbf y\mathbf y^T= \nonumber\\
	\sum_{\mathbf z \in \Phi^n(\mathbf x) }2^{N(\mathbf{z})}\left[
	\begin{array}{cccc}
	\mathbf{z}_1^2&&&\\
	&\mathbf{z}_2^2&&\\
	&&\ddots&\\
	&&&\mathbf{z}_n^2\\
	\end{array}\right]
	\end{eqnarray}
	Note that,   $N(\mathbf{x})=N(\tau(\mathbf{x}))$ for any $\tau \in \Upsilon^n$. Meanwhile, for  $t_i=\sum_{j=1}^{i}a_j$, $i=1,2,\cdots, M$, the set $\{\mathbf z\mid \mathbf z_1=x_{t_i}~\&~ \mathbf z\in \Phi^n(\mathbf x) \}$ has$\displaystyle \frac{(n-1)!a_i}{\prod_{j=1}^{M}a_j!}$ many members. Thus the (\ref{General}) can be computed as following
	\begin{eqnarray}
	\label{General2}
	&&\sum_{\mathbf y \in \Upsilon^n(\mathbf x)}\mathbf y\mathbf y^T \quad\quad\quad\quad\quad\quad\quad\quad\quad\quad\quad\nonumber\\
	&=&2^{N(\mathbf{x})}\sum_{\mathbf z \in \Theta^n(\mathbf x)}\left[
	\begin{array}{cccc}
	\mathbf z_1^2&&&\\
	&\mathbf z_2^2&&\\
	&&\ddots&\\
	&&&\mathbf z_n^2\\
	\end{array}\right]\\
	&=&2^{N(\mathbf{x})}\sum_{i=1}^{M}\frac{(n-1)!a_i}{\prod_{j=1}^{M}a_j!}x_{t_i}^2 \mathbf E_n\nonumber\\
	&=&2^{N(\mathbf{x})}\frac{(n-1)!}{\prod_{j=1}^{M}a_j!}\sum_{i=1}^{M}x_{t_i}^2a_i \mathbf E_n\nonumber\\
	&=&H_n(\mathbf{x}) \mathbf E_n
	\end{eqnarray}
\end{proof}

\section*{Acknowledgment}
The authors would like to thank the associate editor Prof.Saab and the anonymous reviewers for many constructive comments that helped us to clarify the presentation. We are grateful to Mr Yang Wenqiang for his help of the MATLAB programming. NSF of China partially supported this work (No. 11401061, No. 61202131, and No. 61672488), SRF for ROCS, the CAS western light program, CAS Youth Innovation Promotion Association (No. 2015315),  National Key
R$\&$D Program of China (No. 2018YFC0116704), Chongqing Science and Technology Commission projects cstc2014jcsfglyjs0005 and cstc2014zktjccxyyB0031.

\ifCLASSOPTIONcaptionsoff
  \newpage
\fi

\bibliographystyle{IEEEtran}


\end{document}